\documentclass[runningheads]{llncs}
\usepackage{fullpage}
\usepackage[utf8]{inputenc}
\usepackage[english]{babel}

\usepackage{amsmath}

\usepackage{amsthm}
\usepackage{amsfonts}
\usepackage{amssymb}
\usepackage{amsbsy}
\usepackage{mathtools}
\usepackage{graphicx}
\usepackage[square,sort,sectionbib]{natbib}
\usepackage[inline]{enumitem}
\usepackage[vlined,ruled]{algorithm2e}

\SetKw{AAnd}{and}
\SetKw{Or}{or}
\SetKw{Not}{not}
\SetKw{Break}{break}
\SetAlFnt{\small}
\SetAlCapFnt{\small}
\SetAlCapNameFnt{\small}
\SetAlCapHSkip{-1ex}
\IncMargin{-\parindent}
\SetAlgoCaptionSeparator{}
\usepackage{booktabs}
\usepackage{hyperref}
\hypersetup{
    colorlinks,
    linkcolor={rwthpetrol},
    citecolor={rwthviolet},
    urlcolor={rwthblue}
}
\usepackage{cleveref}
\usepackage{xcolor}
\usepackage{tikz}
\usetikzlibrary{arrows,automata,graphs,shapes,calc,decorations.pathreplacing,decorations.pathmorphing,patterns,decorations.text}
\usetikzlibrary{tikzmark,positioning}
\tikzset{->-/.style={decoration={markings,mark=at position .5 with {\arrow{>}}},postaction={decorate}}}
\tikzset{>=latex'}
\tikzset{vertex/.style={draw,circle,inner sep=0pt,minimum size=15pt}}
\tikzset{medium vertex/.style={draw,circle,inner sep=0pt,minimum size=10pt}}
\tikzset{small vertex/.style={draw,circle,inner sep=0pt,minimum size=8pt}}
\tikzset{box/.style args={#1}{draw=#1,rounded corners,thick,fill=none}}
\tikzset{nobbox/.style args={#1}{draw=none,rounded corners,fill=#1}}
\tikzset{item/.style={circle,inner sep=0pt,minimum size=3pt, fill},>=latex'} 
\pgfdeclarelayer{background}
\pgfsetlayers{background,main}
\newcommand{\tube}[3]{
    \node[draw,circle,minimum size=#3,#2] at #1 {};
}
\newcommand{\ubox}[3]{
    \node[draw,minimum size=#3,#2] at #1 {};
}
\usetikzlibrary{arrows}
\usepackage{appendix}
\usepackage{mathdots}
\usepackage{latexsym}
\usepackage{subcaption}
\usepackage{thm-restate}
\usepackage[inline]{enumitem}

\definecolor{rwthblue}{cmyk}{1, .5, 0, 0}
\definecolor{rwthlightblue}{cmyk}{.45, .14, 0, 0}
\definecolor{rwthred}{cmyk}{.15, 1, 1, 0}
\definecolor{rwthgreen}{cmyk}{.7, 0, 1, 0}
\definecolor{rwthorange}{cmyk}{0, .4, 1, 0}
\definecolor{rwthmagenta}{cmyk}{0, 1, .25, 0}
\definecolor{rwthyellow}{cmyk}{0, 0, 1, 0}
\definecolor{rwthpetrol}{cmyk}{1, .3, .5, .3}
\definecolor{rwthturquoise}{cmyk}{1, 0, .4, 0}
\definecolor{rwthmay}{cmyk}{.35, 0, 1, 0}
\definecolor{rwthbordeaux}{cmyk}{.25, 1, .7, .2}
\definecolor{rwthviolet}{cmyk}{.7, 1, .35, .15}
\definecolor{rwthlila}{cmyk}{.6, .6, 0, 0}

\DeclareFontFamily{U}{mathx}{}
\DeclareFontShape{U}{mathx}{m}{n}{<-> mathx10}{}
\DeclareSymbolFont{mathx}{U}{mathx}{m}{n}
\DeclareMathAccent{\widehat}{0}{mathx}{"70}
\DeclareMathAccent{\widecheck}{0}{mathx}{"71}

\newcommand{\level}{\Theta}

\newcommand{\lmin}{\Theta_{\min}}

\newcommand{\Tt}{\mathcal{T}}

\newcommand{\pr}[2]{\langle #1, #2 \rangle}
\newcommand{\bO}{\mathbf{0}}
\newcommand{\rk}{\rho}
\newcommand{\rkmin}{\check{\rho}}
\newcommand{\rkmax}{\hat{\rho}}
\newcommand{\bz}{\mathbf{z}}
\newcommand{\Z}{\mathbb{Z}}
\newcommand{\R}{\mathbb{R}}
\newcommand{\B}{\mathcal{B}}
\newcommand{\Bt}{\mathcal{B^\times}}
\newcommand{\U}{b}
\newcommand{\rng}{[\bO,\U]_\Z}

\newcommand{\DM}{\widecheck{\mathcal{D}}}
\newcommand{\dset}{\mathcal{D}}
\newcommand{\dhat}{\widehat{\mathcal{D}}}
\newcommand{\pmin}{{p_*}}
\newcommand{\pmax}{{p^*}}
\newcommand{\mr}{\theta}
\newcommand{\supp}{\mathrm{supp}}
\newcommand{\I}{\mathcal{I}}
\newcommand{\abs}[1]{\left\lvert #1 \right\rvert}
\newcommand{\Mnat}{\ensuremath{\mathrm{M}^\natural}}
\newcommand{\Lnat}{\ensuremath{\mathrm{L}^\natural}}
\newcommand{\M}{\ensuremath{\mathrm{M}}}
\newcommand{\EO}{\mathrm{ExO}}
\newcommand{\BO}{\mathrm{DO}}

\newcommand{\DRO}{\mathrm{DRO}}
\newcommand{\BB}{{\color{rwthblue}B}}
\newcommand{\RR}{{\color{rwthred}R}}
\newcommand{\GG}{{\color{rwthgreen}G}}
\let\emptyset\varnothing
\let\epsilon\varepsilon
\let\phi\varphi
\DeclareMathOperator{\overd}{overd}
\DeclareMathOperator{\underd}{underd}

\usepackage{tabularx}
\usepackage{ragged2e}
\newcolumntype{C}[1]{>{\centering\arraybackslash}p{#1}}

\newcommand\blfootnote[1]{%
  \begingroup
  \renewcommand\thefootnote{}\footnote{#1}%
  \addtocounter{footnote}{-1}%
  \endgroup
}

\begin{document}

\title{Faster Dynamic Auctions via Polymatroid Sum}

\author{
	Katharina Eickhoff\inst{1}\thanks{KE is funded by the Deutsche Forschungsgemeinschaft (DFG, German Research Foundation) – 2236/2.} \and
	Meike Neuwohner\inst{2}\thanks{MN was supported by the Engineering and Physical Sciences Research Council, part of UK Research and Innovation, grant ref. EP/X030989/1.} \and
	Britta Peis\inst{1} \and
	Niklas Rieken\inst{1} \and%
	\\%
	Laura Vargas Koch\inst{3} \and
	László A. Végh\inst{2}\inst{4}\thanks{LAV received funding from the European Research Council (ERC) under the European Union’s Horizon 2020 research and innovation programme (grant agreement no. ScaleOpt– 757481).}
}

\institute{
   School of Business and Economics, RWTH Aachen University \and
   Department of Mathematics, London School of Economics and Political Science \and
   Hausdorff Center for Mathematics and Research Institute of Discrete Mathematics, University of Bonn \and
   Institute for Advanced Studies, Corvinus University
}

\maketitle
\blfootnote{An abstract of this paper appeared in the proceedings of WINE 2023 \cite{eickhoff2023faster}.}\vspace{-3ex}

\begin{abstract}
    We consider dynamic auctions for finding Walrasian equilibria in markets with indivisible items and strong gross substitutes valuation functions.
    Each price adjustment step in these auction algorithms requires finding an inclusion-wise minimal maximal overdemanded set or an inclusion-wise minimal maximal underdemanded set at the current prices.
    Both can be formulated as a submodular function minimization problem.

    We observe that minimizing this submodular function corresponds to a polymatroid sum problem, and using this viewpoint, we give a fast and simple push-relabel algorithm for finding the required sets.
    This improves on the previously best running time of Murota, Shioura and Yang (ISAAC 2013).
    Our algorithm is an adaptation of the push-relabel framework by Frank and Mikl\'os (JJIAM 2012) to the particular setting.
    We obtain a further improvement for the special case of unit-supplies.

	We further show the following monotonicty properties of Walrasian prices: both the minimal and maximal Walrasian prices can only increase if supply of goods decreases, or if the demand of buyers increases.
	This is derived from a fine-grained analysis of market prices.
	We call \emph{``packing prices''} a price vector such that there is a feasible allocation where each buyer obtains a utility maximizing set.
	Conversely, by \emph{``covering prices''} we mean a price vector such that there exists a collection of utility maximizing sets of the buyers that include all available goods.
	We show that for strong gross substitutes valuations, the component-wise minimal packing prices coincide with the minimal Walrasian prices and the component-wise maximal covering prices coincide with the maximal Walrasian prices.
	These properties in turn lead to the price monotonicity results.
\end{abstract}

\section{Introduction}\label{sec:intro}
We study computational and structural aspects of Walrasian market equilibria.
Let us consider a market involving a set of buyers with quasi-linear utilities that are interested in buying a set of goods.
The main objective is to find an allocation of the goods to buyers along with prices for the goods which should fulfill the following properties:
\begin{enumerate*}[label=(\roman*)]
    \item every buyer gets a preferred bundle, i.e., a set of items maximizing her utility at the given prices, \label{goals:preferred}
    \item no good is oversold, \label{goals:packing}
    \item all goods are sold, and \label{goals:covering}
   \item the allocation is Pareto-optimal.\label{goals:pareto}
\end{enumerate*}
Properties~\ref{goals:preferred}--\ref{goals:covering} together are the \emph{Walrasian equilibrium} conditions, which coincidentally imply property~\ref{goals:pareto}; this is often referred to as \emph{First Welfare Theorem}.
In this paper, we are interested in finding such an equilibrium allocation and prices via a \emph{dynamic auction}.
In simple terms, in a dynamic auction, the auctioneer iteratively announces prices for goods and the buyers report their demand correspondence.
If there is a feasible way to satisfy all the demands (i.e. properties~\ref{goals:preferred} and \ref{goals:packing} are satisfied) and all goods are distributed to the buyers (i.e. property~\ref{goals:covering} is satisfied in addition) then the auction terminates with the announced prices and an allocation that is feasible.
Otherwise, prices are adjusted in an intuitive way: goods that are overdemanded increase in price, whereas underdemanded goods decrease in price.

It is easy to see that such a dynamic auction can take longer than a static \emph{sealed-bid auction}, in which the auctioneer receives all the buyers' information at the start instead of incrementally after announcing new prices.
The are several reasons why a dynamic auction may be preferable (see \cite{ausubel2004efficient, cramton1998ascending} for a detailed discussion).
First of all, dynamic auctions as described above are more transparent; buyers observe how some goods increase or decrease in price, whereas in a sealed-bid auction, the auctioneer just sets a price whose origin remains opaque to the buyers as the auctioneer is the only agent with full information.
Typically, the buyers' valuation functions (and hence, the utilities upon receiving a bundle at a given price) are private information and buyers do not want to disclose this to an auctioneer in full if it is not necessary.
Moreover, while one can argue that sealed-bid auctions might run faster given complete information of buyers (see~\cite{leme2020computing}), this is only true under the assumption that there are no communication costs between buyers and auctioneer; after all, a valuation function of a buyer has to specify a value for every subset of goods (i.e. exponentially many in the number of goods) and all these values have to be communicated to the auctioneer.
In fact, one might even argue that the exact valuation functions might not be known to the buyers themselves if they are only implicitly given in form of some combinatorial structure, so some of the computational tasks are just outsourced to the buyers.

As mentioned above, a dynamic auction iteratively adjusts prices by raising prices on underdemanded, or overdemanded sets.
Overdemanded [underdemanded] sets are those sets of items that minimize [maximize] the demand minus the supply.
We will study the following four types of dynamic auctions (see \Cref{sec:auctions} for a precise description), which are known to terminate with the unique component-wise minimal [maximal] Walrasian price vector in auctions with strong gross substitutes, provided that prices are raised [decreased] on the inclusion-wise minimal maximal overdemanded [underdemanded] set in each iteration.
In this paper, we present fast algorithms for the problem of finding such an inclusion-wise minimal maximal overdemanded [underdemanded] set.

\begin{enumerate}[label={(\Roman*)}]
    \item An ascending auction that finds minimal Walrasian prices from any lower bound on prices,\label{auctions:ascending}
    \item a descending auction that finds maximal Walrasian prices from any upper bound on prices,\label{auction:descending}
    \item a two-phase auction, that finds some Walrasian prices from any initial price vector, and\label{auctions:two-phase}
    \item a greedy auction, that finds some Walrasian prices from any initial price vector.
    \label{auctions:greedy}
\end{enumerate}
The auctions~\ref{auctions:ascending}--\ref{auctions:two-phase} can be viewed as efficient implementations of the corresponding auctions by~\citet{ausubel2006efficient}, whereas the auction~\ref{auctions:greedy} is a speeded-up version of the auction by \citet{murota2016time}.

Moreover, we are able to show that the minimal and maximal Walrasian prices have the following properties, where we refer to the component-wise ordering among price vectors.
\begin{enumerate}
    \item Every price vector satisfying~\ref{goals:preferred} and \ref{goals:packing} is at least as large as the minimal Walrasian price vector,
    \item every price vector satisfying~\ref{goals:preferred} and \ref{goals:covering} is at most as large as the maximal Walrasian price vector, and
    \item minimal and maximal Walrasian price vectors react naturally to changes in demand and supply, i.e. they cannot decrease if demand increases or supply decreases and they cannot increase if demand decreases or supply increases.
\end{enumerate}

To start with, we first have a closer look into the economic model.

\subsection{The economic model}
In this paper, we consider the following market.
Let $E = \{1, \ldots, m\}$ be a set of $m$ \emph{item types} (we will just refer to those as \emph{items}), and $N = \{1, \ldots, n\}$ be a set of $n$ \emph{buyers}.
Each item $e \in E$ is available in $\U(e)$ indivisible units for some $\U \in \Z_+^E$, and each buyer $i \in N$ has an individual valuation function $v_i\colon \rng \to \Z_+$ over the set $\rng \coloneqq \{z \in \Z^m \mid \bO \leq z \leq \U\}$ of possible \emph{bundles}.
For $z \in \rng$, the value $z(e)$ represents the number of items of type $e \in E$.
The value $v_i(z) \in \Z_+$ can be interpreted as buyer $i$'s maximum willingness to pay for a bundle $z \in \rng$.
Throughout the paper, we assume that all valuation functions are \emph{non-decreasing}, i.e. for all $z, z' \in \rng$ with $z \leq z'$, we have $v_i(z) \leq v_i(z')$.
We let $B \coloneqq \max_{e \in E} \U(e)$.
We refer to the case where $\U(e) = 1$ for all $e \in E$ as the \emph{unit-supply} case, and to the general problem as the \emph{multi-supply} case.

For a given price vector $p \in \R_+^E$, the \emph{utility function} of buyer $i$ is given by $u_i(z) = v_i(z) - \pr{p}{z}$.
Buyer $i$'s \emph{indirect utility function} is given by
\begin{equation*}
    V_i(p) \coloneqq \max \left\{v_i(z) - \pr{p}{z} \mid z \in \rng \right\}\, .
\end{equation*}
Let us denote by $\dset_i(p)$ the \emph{set of preferred bundles} (also called \emph{demand set}), which are the maximizers of $u_i$,
that is,
\begin{equation*}
    \dset_i(p) \coloneqq \arg\max\left\{v_i(z) - \pr{p}{z} \mid z \in \rng \right\}\hspace{5ex}.
\end{equation*}
Further let $\DM_i(p)$ denote the \emph{set of minimal preferred bundles} and $\dhat_i(p)$ denote the \emph{set of maximal preferred bundles} of buyer $i$ for a given price vector $p$, i.e.,
\begin{equation*}
    \DM_i(p) \coloneqq \{ z \in \dset_i(p) \mid \text{there is no } y \in \dset_i(p)\setminus \{z\} \text{ with } y \leq z\}
\end{equation*}
and
\begin{equation*}
    \dhat_i(p) \coloneqq \{ z \in \dset_i(p) \mid \text{there is no } y \in \dset_i(p)\setminus \{z\} \text{ with } y \geq z\}.
\end{equation*}

An allocation $(z_1, \dots, z_n)$ is called \emph{packing} with respect to the prices $p \in \R_+^E$ if
\begin{equation}
    \tag{packing}
    z_i \in \dset_i(p) \text{ for all } i \in N \hspace{5ex} \text{ and } \hspace{5ex} \sum_{i \in N} z_i(e) \leq \U(e) \text{ for all } e \in E.
\end{equation}
Packing refers to the fact that one preferred bundle per buyer can be chosen such that the sets can be packed in the set of items (no item is assigned too often).
A price vector $p$ that supports a packing allocation is also called \emph{packing}.\footnote{Note that \cite{ben2013ascending} call packing prices \emph{envy-free prices}.}
Note that it is easy to get packing prices by just setting them sufficiently high such that no buyer is interested in any item anymore.

An allocation $(z_1, \dots, z_n)$ is called \emph{covering} with respect to the prices $p \in \R_+^E$ if
\begin{equation}
    \tag{covering}
    z_i \in \dset_i(p) \text{ for all } i \in N \hspace{5ex} \text{ and } \hspace{5ex} \sum_{i \in N} z_i(e) \geq \U(e) \text{ for all } e \in E.
\end{equation}
Here, covering refers to the fact that there is a preferred bundle per buyer such that all items are covered by these sets.
We may also call a price vector \emph{covering} if there is a supporting covering allocation.

A price vector $p \in \R_+^E$ is called a \emph{Walrasian price vector} \citep{walras1874} or a \emph{competitive equilibrium price vector} if there is an allocation $(z_1, \ldots , z_n)$ with
\begin{equation}
    \tag{Walrasian}
    z_i \in \dset_i(p) \text{ for all } i \in N \hspace{5ex} \text{ and } \hspace{5ex} \sum_{i \in N} z_i(e) = \U(e) \text{ for all } e \in E.
\end{equation}
So there is one allocation that is both packing and covering.
We show in \Cref{apx:packing_covering} that in auctions with strong gross substitutes valuations, a price vector is Walrasian if and only if it is both packing and covering, i.e., that there is an allocation with both properties and not just different allocations, one packing and one covering.

Walrasian prices may not always exist (see Example~\ref{example:nongs_nowalrasian} in \Cref{app:Walrasian_prices_dont_exist_without_GS}).
However, Walrasian prices are guaranteed to exist if all valuation functions of the buyers satisfy the \emph{(strong) gross substitutes}
property \cite{kelso1982job} (see \Cref{def:GS}).
Intuitively, gross substitutes expresses the property that the demand of an item does not decrease if only the prices on other items go up.
For strong gross substitutes valuations, the Walrasian price vectors form a lattice w.r.t.\ the component-wise ordering, see \citet{gul1999walrasian} and \citet{ausubel2006efficient}.
Consequently, there exists a unique component-wise minimal and a unique component-wise maximal Walrasian price vector.
We also call the component-wise minimal Walrasian prices \emph{buyer-optimal Walrasian prices} and the component-wise maximal Walrasian prices \emph{seller-optimal Walrasian prices}.

In this paper, we focus on computational aspects of the ascending auctions described by \citet{gul2000english} (for unit-supply) and by \citet{ausubel2006efficient} (for multi-supply) for computing the minimal [maximal] Walrasian price vector in a market with (strong) gross substitutes valuations.
In both cases, the maximal overdemanded [underdemanded] sets are the minimizers of a submodular function.
Thus, the minimal maximal overdemanded [underdemanded] set can be found by general submodular function minimization algorithms.

An exciting connection between market equilibrium computation and discrete optimization was established by \citet{fujishige2003note}, who pointed out the equivalence of the gross substitutes property and \Mnat-concavity, a key concept in the field of \emph{discrete convex analysis}.
This theory combines convex analysis and matroid theory, summarized in the monograph by \citet{murota2003book}.
Thus, powerful tools from discrete convex analysis became available to obtain more efficient equilibrium computations.
In particular, \citet{murota2013computing} obtained faster algorithms for the minimal maximal overdemanded [underdemanded] set computation by exploiting the special structure of the submodular function and taking advantage of demand oracle queries; this will be discussed in Section~\ref{sec:related-work}.

\subsection{Our contributions}
Our contributions are twofold: we provide faster auction algorithms, and we present structural results on price monotonicity.

In the same vein as \citet{murota2013computing}, we exploit the special structure of the \emph{Lyapunov function} to obtain even faster running times.
We view the submodular minimization problem from a dual viewpoint as a \emph{matroid union problem} (for the unit-supply case) or a \emph{polymatroid sum problem} (for the multi-supply case),
and use a fast (and also simple) push-relabel algorithm for this problem.
Let us denote by $\BO$ the time needed for each \emph{demand oracle query} and by $\EO$ the time of each \emph{exchange oracle query} (see Section~\ref{sec:oracle} for definitions and discussion).
Then, our algorithm runs in time $\mathcal{O}(n \cdot \BO + nm^3 \cdot \EO)$ in the multi-supply case.
In comparison, the algorithm in \citep{murota2013computing} runs in time $\mathcal{O}(n \cdot \BO + nm^4 \log(nmB) \cdot \EO)$.
For the unit-supply case, the problem becomes the classical \emph{matroid union problem}, and we obtain an even better running time of $\mathcal{O}(n \cdot \BO + (m^3+nm^2) \cdot \EO)$.

We note that our algorithm is significantly simpler than the one in \citep{murota2013computing}.
We also show that given an optimal solution to the polymatroid sum problem, the minimal maximal overdemanded [underdemanded] set can be found by a simple breadth-first search in an exchange graph.

Our push-relabel algorithm is a special implementation of the more general submodular flow feasibility algorithm by \citet{frank2012simple}.
In this context, our contribution is giving an efficient implementation in terms of the number of oracle calls.
The description in \citep{frank2012simple} is generic and gives bounds in terms of \emph{basic operations}.
Implementing a single such operation may take $\mathcal{O}(nm \cdot \EO)$; however, we show that this can be amortized over a sequence of basic operations.
See Section~\ref{sec:related-work} on further discussion of push-relabel algorithms.

Our second main contribution is on price monotonicity.
We make a clear conceptual distinction between packing, covering and Walrasian prices.
For example, it is not clear a priori whether there may exist a packing price vector $q$ with $q(e) < \pmin(e)$ for some items $e \in E$, where $\pmin$ denotes the buyer-optimal Walrasian price vector, or whether there may exist a covering price vector $q'$ with $q'(e) > \pmax(e)$ for some items $e \in E$, where $\pmax$ is the maximal Walrasian price vector.
In \Cref{thm:min_packing_is_min_walrasian}, we show that in fact $\pmin \leq q$ for all packing price vectors $q$ and in \Cref{thm:max_market_clearing_is_max_walrasian}, we show that $\pmax \geq q'$ for all covering price vectors $q'$.

Building on this result, we can also prove in \Cref{sec:sensitivity} that the minimal [maximal] Walrasian prices react naturally to changes in supply and demand, i.e., if total supply of items decreases or the total demand of buyers increases, the minimal [maximal] Walrasian prices can only increase.
Independently, \citet{raach2024monotonicity} proved the same monotonicity results.

\subsection{Related work}\label{sec:related-work}

In this section, we give a literature review over the topics related to our economic model and equilibrium concept.
For an excellent overview over (strong) gross substitutes, Walrasian equilibria, and ascending auctions we refer to the survey by \citet{leme2017gross}.

\paragraph{Gross substitutes and Walrasian equilibria}
\citet{kelso1982job} showed in their seminal paper that a Walrasian equilibrium is guaranteed to exist if all valuation functions satisfy the gross substitutes condition that in layman's terms can be stated as, \textit{``the demand for an item does not not decrease, if only the prices of other items are increased''}.\footnote{A formal definition is provided in Section~\ref{sec:preliminaries}.}
Roughly speaking, the study of economic models from an algorithmic and complexity theoretical point of view really started off in the 1990s when also the foundations of algorithmic game theory were laid out.
\citet{gul1999walrasian} showed with their Maximum Domain Theorem that indeed gross substitutes are the largest class of valuation functions (containing unit-demand valuations) that guarantee the existence of Walrasian equilibria.
Additionally, they provided equivalent definitions for the gross substitutes condition and showed that Walrasian prices form a complete lattice, which implies that there exist unique component-wise minimal and maximal Walrasian price vectors.
There are additional equivalent characterization of gross substitutes in \cite{ben2013ascending}.
However, the characterization by \citet{fujishige2003note} turns out to be most useful from a mathematical point of view.
It states that a valuation function has the gross substitutes property if and only if it is \Mnat-concave, which allows for the usage of powerful tools from discrete convexity.
This result also has been extended to strong gross substitutes valuation functions in \cite{murota2013computing}.
The differences between strong gross substitutes and ``weak gross substitutes'' is discussed in \cite{milgrom2009substitute}.

There are other classes than (strong) gross substitutes for which Walrasian prices are guaranteed to exist \cite{sun2009double, ben2013ascending} but these classes do not contain the natural unit-demand valuations.

\paragraph{Dynamic Auctions or Walrasian t\^atonnement}
L\'eon Walras, the namesake of our equilibrium concept, already proposed how equilibria may be found, namely by a `t\^atonnement'\footnote{French: trial-and-error} process.
This procedure basically describes a dynamic auction: An auctioneer posts prices and unless these prices are at equilibrium, the auctioneer makes an adjustment.
The first modern study of such a process was done by \citet{demange1986} for unit-demand valuations.
In Gul and Stacchetti's follow-up work \cite{gul2000english}, they gave the framework for ascending auctions for gross substitutes valuations; that is, start at all-zero prices and increase prices on an inclusion-wise minimal maximal overdemanded set until there is no overdemanded set anymore.
They proved that such an auction always terminates with minimal Walrasian prices.
However, while they showed that an overdemanded set of items has to exist whenever the prices are not yet at equilibrium, they left it open how to compute those sets.
This gap was initially closed by \citet{ausubel2006efficient,ausubel2005walrasian} using submodular function minimization.
There is a discrete version of the Lyapunov function (a potential function, introduced by \citet{varian1981dynamical} for divisible goods) which is minimized at Walrasian prices \cite{ausubel2006efficient}.
If all valuation functions are gross substitutes, then this function is submodular and its minimum (and every steepest descent direction, which corresponds to a maximal overdemanded set) can be found efficiently using submodular function minimization \cite{grotschel1981ellipsoid,iwata2001combinatorial,iwata2003faster,iwata2009simple}.
\citet{murota2013computing} showed that the Lyapunov function is not just submodular but \Lnat-convex, which allows for even faster methods than plain submodular function minimization.
Note that these methods by \citet{ausubel2006efficient} and \citet{murota2013computing} allow for descending auctions or other kinds of dynamic auctions, where prices maybe adjusted in a non-monotone fashion.
Similar auctions and guarantees for termination with minimal/maximal Walrasian prices are discussed in \cite{andersson2013sets} and \cite{ben2013ascending}.

We should mention that the literature mentioned above follows the definition of a dynamic auction as proposed by \citet{gul2000english}.
However, the literature also explores other auction designs, e.g. in \cite{bikhchandani2002package,ausubel2004efficient,parkes2000iterative,ausubel2002ascending} that consider slightly different market settings or different ways for buyers to report their demands to the auctioneer.

A weak point of dynamic auctions would be their running time if we assume full information for the auctioneer; e.g., in an ascending auction $\|\pmin\|_\infty$ price increase steps are needed, where $\pmin$ is the (minimal) Walrasian price vector computed by the auction \cite{murota2013computing}.
This process can be sped up by increasing prices on an overdemanded set not only by one but by the maximal possible amount before the steepest descent direction of the Lyapunov function changes.
This results in at most $nmB$ price adjustment steps (c.f.\ \citep[Proposition 4.17]{shioura2017algorithms}).
Further speed ups of the ascending auction as a rounding scheme are discussed in \citep[Section 10.1.]{leme2017gross}.
There are other algorithms which can compute Walrasian prices efficiently \cite{leme2020computing}.
However, in contrast to dynamic auctions, these algorithms need direct information of the valuation functions, i.e., by a value oracle.

\paragraph{Oracles and communication complexity}
To compute Walrasian prices, it is crucial to get some information on buyers' valuation functions.
Usually, the necessary communication between auctioneer and buyers is modeled via oracle calls.
Depending on which oracle type is used, one can get different information.
A comparison of the computational power of some oracle types can be found, e.g., in \cite{blumrosen2010computational}.
In auction literature, typically, value oracles and demand oracles are used.
For our dynamic auctions, we require a demand oracle and an exchange oracle as those work naturally in our proposed algorithm to find the required overdemanded or underdemanded sets.
The exchange oracle is related to the dynamic rank oracle \cite{blikstad2023fast}, as both provide information on sets close to already requested ones.
A detailed overview which oracles we use and how they can be compared is given in \Cref{sec:oracle}.

\paragraph{VCG prices.}
As a corner stone of auction theory and mechanism design, we have the VCG Theorem \cite{vickrey1961counterspeculation,clarke1971multipart,groves1973incentives} that states that there is a sealed-bid mechanism that has truthful bidding as a dominant strategy.
However, as briefly discussed in the introduction, ascending auctions are often preferred as they yield a more natural and transparent price-finding process.
As a consequence, mathematical economists went to great lengths to implement the VCG mechanism in an ascending auction \cite{ausubel2004efficient,de2007ascending,bikhchandani2001linear,bikhchandani2011ascending,mishra2007ascending}.
\citet{kern2016} showed that for unit-demand valuations, the auction by \citet{demange1986} mentioned above indeed returns VCG prices for unit-demand valuations.
If we go beyond unit-demand valuations, VCG prices cannot be achieved, as shown by \citet{gul2000english}, as the ascending auction does not gain enough information during the iterations to determine VCG prices.
In a multi-supply setting, \cite{eickhoff2023flow} show that VCG prices that are also Walrasian do not necessarily exist.

\paragraph{Discrete convexity: matroid union and polymatroid sum}
We refer the interested reader to the monographs by \citet{murota2003book} and \citet{fujishige2005submodular}.
A survey on the connections between discrete convexity and auction theory is \cite{shioura2015gross}.

As pointed out above, all buyers having (strong) gross substitutes valuation functions means that we can use methods from discrete convex optimization to solve the computational problem that arises during a dynamic auction.
In particular, the whole auction can be viewed as an iterative process to minimize a submodular (and even \Lnat-convex) function \cite{ausubel2006efficient,murota2013computing,murota2016time,shioura2017algorithms}.
As we will see in the sequel,
these methods can be quite complicated.
For the unit-supply case and gross substitutes valuations, the problem of finding a maximal overdemanded set becomes equivalent to the matroid union problem on $n$ matroids, which is efficiently solvable as shown in \citep{edmonds1968matroid}.
Over time, more efficient algorithms were developed \citep{knuth1973matroid,greene1975some,cunningham1986improved,chakrabarty2019faster} and most recently by \citet{terao2023faster} and \citet{blikstad2023fast}.
Note that some of these methods are hard to compare due to the use of different oracles.

In the more general setting with multi-supply and strong gross substitutes valuations, we need to consider the more general polymatroid sum problem.

\paragraph{Push-relabel algorithms}
The algorithms mentioned above to solve the matroid union and the polymatroid sum problem are so called \emph{augmenting path algorithms}.
They start by taking a trivially feasible solution (e.g. the empty set or the all-zero vector, respectively) and augment it by finding additional elements that can be added to the current solution or $(k+1)$-to-$k$-swaps to increase the size of the solution until it is optimal.
However, more recently, so called \emph{push-relabel algorithms} have been studied for various combinatorial optimization problems and they can outperform augmenting path algorithms both in theory as well as in practice.
The push-relabel paradigm was introduced for the maximum flow problem by \citet{goldberg1988new}.
The idea is in some sense dual to augmenting path algorithms: One starts with a possibly infeasible solution which otherwise satisfies some optimality condition.
The push and relabel operations then maintain this optimality condition while making the solution ``less infeasible'' over time.
We use a
 push-relabel algorithm by \citet{frank2012simple} for matroid union.
Our algorithm in Section~\ref{sec:push-relabel} gives an efficient implementation in the auction setting and also generalizes to polymatroid sum and \M-convex sets in general.
It is interesting to note that the level functions and local exchanges of push-relabel algorithms have some resemblance to auction algorithms, as already pointed out by \citet{bertsekas1992auction}.
In our push-relabel algorithm in Section~\ref{sec:push-relabel}, one could interpret the level $\level(e)$ of an item as a small marginal price discount used for tie-breaking.

\subsection{Overview}
In \Cref{sec:preliminaries}, we start by giving some basic definitions (\Cref{sec:gs}) and the well-known dynamic auctions that iteratively increases [decrease] prices on items that are overdemanded [underdemanded] (\Cref{sec:auctions}).
This is followed by a summary of known definitions and facts from discrete convex analysis and {(poly-)}{matroids} in \Cref{sec:discrete-convexity}.
Finally, we discuss oracle models in \Cref{sec:oracle}.

In \Cref{sec:finding_over_and_underdemanded_sets}, we describe how to find minimal maximal overdemanded/underdemanded sets, given an optimal solution to the corresponding polymatroid sum problem.
Afterwards, in \Cref{sec:push-relabel}, we present the push-relabel framework and show how to implement it in an efficient way for both the matroid union and the polymatroid sum problem.
Combining these results enables us to compute the minimal maximal overdemanded [underdemanded] sets fast, while using only oracle calls that are natural in the auction setting.
In \Cref{sec:minimal-packing}, we show that there exists a component-wise minimal packing price vector and that it coincides with the minimal Walrasian price vector.
Further, we show that there also exists a component-wise maximal covering price vector, which is equal to the maximal Walrasian price vector.
In \Cref{sec:sensitivity}, we use these facts to show that Walrasian prices fulfill natural monotonicity properties, i.e., when the supply decreases or the demand increases, Walrasian prices can only increase.

\section{Preliminaries}\label{sec:preliminaries}
In this section, we give some basic definitions and known facts on the connection between strong gross substitutes valuations, polymatroids and discrete convex optimization.
We introduce the ascending auction presented by \citet{gul2000english} for the unit-supply setting, which was extended by \citet{ausubel2006efficient} for the multi-supply setting.

\paragraph{Notation} For a vector $x \in \R^E$, we let $\supp^{+}(x) \coloneqq \{e \in E \mid x(e) > 0\}$ and $\supp^{-}(z) \coloneqq \{e \in E \mid x(e) < 0\}$.
For $S \subseteq E$, we let $x(S) \coloneqq \sum_{e \in S} x(e)$.
For $x, y \in \R^E$, we defined the meet $x \wedge y \in \R^E$ as the component-wise minimum, and the join $x \vee y \in \R^E$ as the component-wise maximum of $x$ and $y$.

\subsection{Strong gross substitutes valuations}
\label{sec:gs}
Recall from the introduction that $\dset_i(p)$ is the set of preferred bundles of buyer $i \in N$ with respect to prices $p$.
\citet{kelso1982job} introduced the concept of gross substitutes valuations which guarantee the existence of Walrasian prices.
Intuitively, gross substitutes expresses the property that the demand for an item will not decrease if only prices of other items are raised.

\begin{definition}
    A valuation function $v_i\colon 2^E \to \Z$ is \emph{gross substitutes} if for all price vectors $p, q \in \R^E$ with $p \leq q$, it holds that for all $x_i \in \dset_i(p)$, there exists $y_i \in \dset_i(q)$ such that $x_i(e) \leq y_i(e)$ for all $e \in E$ with $p(e) = q(e)$.
\end{definition}

Note that this definition of gross substitutes (see \cite{kelso1982job}) only considers the unit-supply case (i.e. where $\U(e) = 1$ for all $e \in E$).
It can be generalized to the multi-supply setting by treating each copy of an item type as an individual item \cite{ausubel2006efficient}.
The following direct definition (see \cite{milgrom2009substitute}) is a bit more involved.
In particular, it includes the additional requirement that higher prices cannot lead to larger preferred bundles.

\begin{definition}\label{def:GS}
    A valuation function $v_i\colon \rng \to \Z$ is \emph{strong gross substitutes} if for all price vectors $p, q \in \R^E$ with $p \leq q$, it holds that for all $x_i \in \dset_i(p)$, there exists $y_i \in \dset_i(q)$ such that $x_i(e) \leq y_i(e)$ for all $e \in E$ with $p(e) = q(e)$ and such that $x_i(E) \geq y_i(E)$.
\end{definition}

Strong gross substitutes valuations form a natural and well-studied class of valuation functions.
Such functions are always submodular, i.e., have the diminishing marginal returns property (see Definition~\ref{def:submodular} below).
Simple examples of strong gross substitutes valuations include unit-demand valuations (i.e., $v(z) = \max_{e \in \supp^+(z)} w(e)$ for $w \in \R^m_+$), weighted matroid and polymatroid rank functions (i.e., $v(z) = \max \{\pr{w}{y} : y \leq z, y(S) \leq \rk(S), S \subseteq E\}$ for $w \in \R^m_+$), and OXS valuations (i.e., the maximum weight of a matching in a bipartite graph, restricted to a subset of the items).

It is well-known that strong gross substitutes are the largest class of functions (containing the unit-demand valuations) that guarantee existence of Walrasian equilibria \citep[Theorem~2]{gul1999walrasian}.

In the sequel, we assume that all valuation functions are strong gross substitutes and non-decreasing.
In this case, it is known that the collection of minimal or maximal preferred bundles forms an \M-convex set (see \Cref{sec:discrete-convexity}).
This allows us to use tools from discrete optimization and integral polymatroids.

\subsection{Dynamic auctions}
\label{sec:auctions}
In this section, we present multiple dynamic auctions that can be used to compute Walrasian prices, if they exist.
All of them crucially rely on the observation that a price vector $p$ is Walrasian if and only if it constitutes a minimizer of the Lyapunov function \cite{ausubel2006efficient, kelso1982job}
\[
L(p) = \sum_{i \in N} V_i(p) + \pr{p}{\U}.
\]
\begin{lemma}[\cite{ausubel2006efficient, kelso1982job}] \label{lem:lyapunov}
	If all buyers have strong gross substitutes valuation functions, then prices $p$ are Walrasian if and only if they minimize the Lyapunov function.
\end{lemma}

In light of \Cref{lem:lyapunov}, a natural approach to compute a Walrasian price vector is to start with an arbitrary price vector $p$ and to iteratively replace $p$ by $p\pm\chi_S$ for some $S\subseteq E$ if this decreases the value of the Lyapunov function.

It turns out that the sets $S\subseteq E$ minimizing $L(p+\chi_S)$ [$L(p-\chi_S)$]
are precisely the maximal overdemanded [underdemanded] sets.
For the unit-supply case, this is shown in~\cite{ben2017walrasian}.
With a standard copy trick, we can obtain the same result for the multi-supply case, see \Cref{apx:missing-proofs:max-overd-steep}.

\citet{ausubel2006efficient} shows that an auction, starting with an arbitrary price vector $p$ and then alternating increasing and decreasing price adjustments, will return a Walrasian price vector.
Further, \citet{sun2009double} show that by first iteratively increasing prices on a minimal maximal overdemanded set, and then iteratively decreasing prices on a minimal maximal underdemanded set, the auction terminates with a Walrasian price vector as well.
This is what we call a \ref{alg:two-phase_auction}.

\begin{algorithm}[H]
	\SetAlgoRefName{Two-Phase Auction}
	Start with an arbitrary price vector $p$\\
	\While{an overdemanded set exists}{
		choose minimal maximal overdemanded set $S \subseteq E$\\
		$p(e) \coloneqq p(e) + 1$ for all $e \in S$
	}
	\While{an underdemanded set exists}{
		choose minimal maximal underdemanded set $S \subseteq E$\\
		$p(e) \coloneqq p(e) - 1$ for all $e \in S$
	}
	\Return{$p$}
	\caption{}
	\label{alg:two-phase_auction}
\end{algorithm}
\citet{murota2013computing} shows that one can reduce the number of iterations by interleaving both phases and always performing an increasing or decreasing step depending on which of the two provides the greater reduction of the Lyapunov function.
This results in the following \ref{alg:greedy_auction}.

\begin{algorithm}[H]
	\SetAlgoRefName{Greedy Auction}
	Start with an arbitrary price vector $p$\\
	\While{an overdemanded or underdemanded set exists}{
		choose minimal maximal overdemanded set $S_o \subseteq E$\\
		choose minimal maximal underdemanded set $S_u \subseteq E$\\
		\uIf{$\overd(S_o) \geq \underd(S_u)$}{
			$p(e) \coloneqq p(e) + 1$ for all $e \in S_o$
		}
		\Else{
			$p(e) \coloneqq p(e) - 1$ for all $e \in S_u$
		}
	}
	\Return{$p$}
	\caption{}
	\label{alg:greedy_auction}
\end{algorithm}
As the Lyapunov function is submodular and the minimizers of any submodular function form a lattice, by \Cref{lem:lyapunov}, this is also true for the Walrasian price vectors.
In particular, there exists a \emph{unique} buyer-optimal, i.e., minimal [seller-optimal, i.e., maximal] Walrasian price vector $\pmin$ [$\pmax$].
\citet{ausubel2006efficient} has shown that $\pmin$ [$\pmax$] can be obtained via the following \ref{alg:generic_ascending_auction} [\ref{alg:generic_descending_auction}].

\begin{algorithm}[H]
    \SetAlgoRefName{Ascending Auction}
    $p(e) \coloneqq 0$ for all $e \in E$\\
    \While{an overdemanded set exists}{
        choose minimal maximal overdemanded set $S \subseteq E$\\
        $p(e) \coloneqq p(e) + 1$ for all $e \in S$
    }
    \Return{$\pmin \coloneqq p$}
    \caption{}
    \label{alg:generic_ascending_auction}
\end{algorithm}
\begin{algorithm}[H]
	\SetAlgoRefName{Descending Auction}
	$p(e) \coloneqq \max_{i\in N}v_i(E)+1$ for all $e \in E$\\
	\While{an underdemanded set exists}{
		choose minimal maximal underdemanded set $S \subseteq E$\\
		$p(e) \coloneqq p(e) - 1$ for all $e \in S$
	}
	\Return{$\pmax \coloneqq p$}
	\caption{}
	\label{alg:generic_descending_auction}
\end{algorithm}
Note that both the \ref{alg:generic_ascending_auction} and the \ref{alg:generic_descending_auction} constitute special cases of both the \ref{alg:two-phase_auction} and the \ref{alg:greedy_auction}: When starting with $p(e) \coloneqq 0$ for all $e \in E$ [$p(e) \coloneqq \max_{i\in N}v_i(E)+1$ for all $e \in E$], throughout the algorithm, there will never be an underdemanded [overdemanded] set (see, e.g.,~\cite{ben2017walrasian}).

We further point out that for the \ref{alg:generic_ascending_auction} [\ref{alg:generic_descending_auction}] to return the buyer-optimal [seller-optimal] Walrasian price vector,
it suffices to start with a vector $\leq \pmin$ or $\geq \pmax$, respectively (see \cite{ausubel2004efficient}).
The running time of the above auctions depends crucially on the time required to find a minimal maximal over- or underdemanded set in each iteration.
In this paper, we show how to obtain faster algorithms to compute minimal maximal over- and underdemanded sets by leveraging an algorithm for the polymatroid sum problem.

\subsection{Discrete convex optimization and polymatroids}\label{sec:discrete-convexity}
In this section,
we summarize some basics from discrete convex optimization and their connection to polymatroids.
In the following sections, we show that in markets with strong gross substitutes valuations, the set of minimal preferred bundles $\DM_i(p)$, and also the set of maximal preferred bundles $\dhat_i(p)$, form the sets of bases of integral polymatroids.
As we will see in \Cref{sec:finding-min-max-overdemanded-set},  we can use the properties stated here to compute the sets on which we increase or decrease the prices in the auction using a polymatroid sum algorithm.

\begin{definition}\label{def:submodular}
    A set function $f\colon 2^E \to \Z$ is \emph{submodular} if
    \[
        f(S) + f(T) \geq f(S \cup T) + f(S \cap T) \quad \text{for all } S, T \subseteq E.
    \]

    A function $v\colon \rng \to \Z$ is \emph{(lattice) submodular} if
    \[
        v(x) + v(y) \geq v(x \vee y) + v(x \wedge y) \quad \text{for all } x, y \in \rng.
    \]
\end{definition}

A special class of submodular functions are \Mnat-concave functions, see \citet{murota2003book}.

\begin{definition}\label{def:Mnat_concave}
    A valuation function $v\colon \rng \to \Z$ is \emph{\Mnat-concave}, if for any $x, y \in \rng$, and for any $e \in \supp^+(x-y)$, one of the following holds:
    \begin{enumerate}[label=(M\arabic*),leftmargin=*]
        \item\label{i:M1} $v(x) + v(y) \leq v(x-\chi_e) + v(y+\chi_e)$, or
        \item\label{i:M2} there exists $f \in \supp^-(x-y)$ such that $v(x) + v(y) \leq v(x-\chi_e+\chi_f) + v(y+\chi_e-\chi_f)$.
    \end{enumerate}
\end{definition}

\begin{lemma}[\citet{murota2013computing}]\label{lemma:gs_is_Mnat-convave}
    A valuation function is strong gross substitutes if and only if it is \Mnat-concave.
\end{lemma}

The utility function $u_i$ of a buyer $i$ is defined as the valuation minus the price, i.e. $u_i(\bz)\coloneqq v_i(\bz) -\pr{p}{\bz}$ for some $\bz \in \rng$ and some given price vector $p$.
Note that the above lemma also implies that a buyer's utility function $u_i$ is \Mnat-concave if her valuation function $v_i$ is strong gross substitutes.
This gives us some nice structure of the minimal and maximal preferred bundles.

\begin{definition}
    A set of bundles $\B \subseteq \rng$ is \emph{\M-convex} if for any $x, y \in \B$, and for any $e \in \supp^+(x-y)$, there exists $f \in \supp^-(x-y)$ such that $x-\chi_e+\chi_f \in \B$ and $y+\chi_e-\chi_f \in \B$.
\end{definition}

\begin{restatable}{lemma}{LemGSMconv}
\label{lem:GS-Mconv}
    If $v_i$ is strong gross substitutes, then $\DM_i(p)$ and $\dhat_i(p)$ are \M-convex sets for any price vector $p \in \R^E_+$.
\end{restatable}
The simple proof of \Cref{lem:GS-Mconv} is given in \Cref{apx:missing-proofs:GS-Mconv}.

As $\DM_i(p)$ and $\dhat_i(p)$ are \M-convex, it follows (see \Cref{apx:missing-proofs:GS-Mconv}) that
\begin{equation*}
    \DM_i(p) = \arg\min\{\|z\|_1 \mid z \in \dset_i(p)\}\hspace{5ex} \text{ and}\hspace{5ex}
    \dhat_i(p) = \arg\max\{\|z\|_1 \mid z \in \dset_i(p)\}.
\end{equation*}
An \M-convex set forms the set of integer points in a corresponding polymatroid base polytope as the next lemma shows.
This leads to a new interpretation of $\DM_i(p)$ and $\dhat_i(p)$.
\begin{lemma}[{\citep[Theorem 4.15]{murota2003book}}]\label{lemma:rk-demand_sets}
    For an \M-convex set $\B \subseteq \rng$, we define
    \begin{equation}\label{eq:rank-def}
        \rk(S) \coloneqq \max\{z(S) \mid z \in \B\} \qquad \text{ for } S \subseteq E.
    \end{equation}
    Then, $\rk(S)$ is an integer valued submodular set function, and $\B$ is the set of integer points in the corresponding base polytope, that is,
   \begin{equation}\label{eq:base-polytope}
        \B = \left\{z \in \Z^m_+ \mid z(S) \leq \rk(S) \text{ for all } S \subseteq E, z(E) = \rk(E)\right\}.
    \end{equation}
\end{lemma}

Summarizing, we observed that if $v_i$ is strong gross substitutes, then for every price vector $p$, buyer $i$'s set of minimal preferred bundles $\DM_i(p)$ as well as the set of maximal preferred bundles~$\dhat_i(p)$ is an \M-convex set.
Moreover, $\DM_i(p)$ is the set of integer points in a polymatroid base polytope and the same holds for $\dhat_i(p)$.

\begin{definition}
    For an \M-convex set $\B \subseteq \rng$ with associated rank function $\rk$ defined as in \eqref{eq:rank-def}, and a vector $z \in \B$, we say that set $S \subseteq E$ is \emph{tight} if $z(S) = \rk(S)$.
\end{definition}
A set of item types $S$ is tight with respect to a minimal [maximal] preferred bundle $z$ of (individual) items if among all minimal [maximal] preferred bundles, the bundle $z$ contains as many items as possible from $S$.

Note that the collection of tight sets $\Tt_{\B}(z) \coloneqq \{S \subseteq E \mid z(S) = \rk(S)\}$ with respect to $z \in \B$ is closed with respect to intersection and union.
To see this, observe that for any two tight sets $S, T \in \Tt_{\B}(z)$, we have $z(S) + z(T) = \rk(S) + \rk(T) \geq \rk(S \cup T) + \rk(S \cap T) \geq z(S \cup T) + z (S \cap T) = z(S) + z(T)$ by submodularity of $\rk$ and since $z \in \B$.
Hence, we get $\rk(S \cup T) = z(S \cup T)$ and $\rk(S \cap T) = z(S \cap T)$.
As a consequence, there exists a \emph{unique} minimal tight set among those tight sets containing $e$, which we call $\Tt_{\B}(e, z)$.

\begin{definition}
    Let $\B$ be an \M-convex set and $z \in \B$.
    We denote the unique minimal tight set containing $e$ by $\Tt_{\B}(e, z)$, i.e.,
    \[
        \Tt_{\B}(e, z) \coloneqq \bigcap \{S \in \Tt_{\B}(z) \mid e \in S\}.
    \]
\end{definition}
We will show later in \Cref{lem:tight-sets-weights} that the set $\Tt_{\DM_i}(e, z_i)$ contains exactly those items $f \in E$ which can be exchanged against $e$, i.e., those $f$ for which $z_i + \chi_e - \chi_f$ remains a minimal preferred bundle for buyer~$i$.
Analogously, $\Tt_{\dhat_i}(e, z_i)$ contains those items $f \in E$ which can be exchanged against $e$ such that $z_i + \chi_e - \chi_f$ remains a maximal preferred bundle.

Let $\B_1, \B_2, \ldots, \B_n \subseteq \rng$ be \M-convex sets, and let $\Bt \coloneqq \bigtimes_{i \in N} \B_i$ be the collection of all $n$-tuples $\bz = (z_1, \ldots, z_n)$ with $z_i \in \B_i$.
The following Min-Max Theorem holds for general polymatroids and implies that prices are packing if and only if there is no overdemanded set (see \Cref{lem:packing}) and that they are covering if and only if there is no underdemanded set (\Cref{lem:covering}).
The Min-Max Theorem follows, e.g., from \citep[(44.8), (44.9), Theorem 44.6]{schrijver2003combinatorial}; the sum of polymatroids was first studied by \citet{mcdiarmid1975rado}.

\begin{theorem}[Min-Max Theorem for Polymatroid Sum]\label{thm:poly-union}
    Let $\B_1, \B_2, \ldots, \B_n \subseteq \rng$ be \M-convex sets with associated rank functions $\rk_i$, $i \in N$ and let $\Bt \coloneqq \bigtimes_{i \in N} \B_i$.
    Then
    \begin{equation}\label{eq:max-min}
        \max_{\bz \in \Bt} \left\{\sum_{e \in E} \min \left\{\sum_{i \in N} z_i(e), \U(e)\right\}\right\}
        = \min_{S \subseteq E} \left\{\sum_{i \in N} \rk_i(E \setminus S) + \U(S)\right\}.
    \end{equation}
\end{theorem}

The problem of finding an optimal solution $\bz \in \Bt$ to the left hand side
\begin{equation}\label{eq:poly-union}
        \max_{\bz \in \Bt} \left\{\sum_{e \in E} \min \left\{\sum_{i \in N} z_i(e), \U(e)\right\}\right\}
    \end{equation}
is also known as the polymatroid sum problem.
The push-relabel algorithm to solve this problem is given in Section~\ref{sec:push-relabel} and provides a direct proof of the Min-Max Theorem for polymatroid sum.

The easy direction $\max \leq \min$ is shown in Lemma \ref{lem:poly-union-opt} below.
Moreover, the lemma also formulates the complementary slackness conditions that can be used to certify optimality of a pair of solutions $\bz \in \Bt$ and $S \subseteq E$.

\begin{lemma}\label{lem:poly-union-opt}
    Let $\B_1, \B_2, \ldots, \B_n \subseteq \rng$ be \M-convex sets with associated rank functions $\rk_i$, $i \in N$.
    Then, for any $\bz = (z_1, \ldots, z_n) \in \Bt \coloneqq \bigtimes_{i \in N} \B_i$, and any set $S \subseteq E$, we have
    \begin{equation}
        \label{eq:poly-union-leq}
        \sum_{e \in E} \min \left\{\sum_{i \in N} z_i(e), \U(e)\right\} \leq \sum_{i \in N} \rk_i(E \setminus S) + \U(S).
    \end{equation}
    Moreover, equality holds if and only if
    \begin{alignat}{3}
        \sum_{i \in N} z_i(e) &\geq \U(e) &&\text{for all } e \in S,\label{eq:sat-E-S}\\
        \sum_{i \in N} z_i(e) &\leq \U(e) &&\text{for all } e \in E \setminus S\,\label{eq:no-overload-S},\\
        z_i(E \setminus S) &= \rk_i(E \setminus S) \qquad &&\text{for all } i \in N.\label{eq:span-S}
    \end{alignat}
\end{lemma}
\begin{proof}
    Let $S \subseteq E$ be an arbitrary set.
    We have
    \begin{equation*}
        \sum_{i \in N} \rk_i(E \setminus S) + \U(S) \geq \sum_{i \in N} z_i(E \setminus S) + \U(S)
        = \sum_{e \in E \setminus S} \sum_{i \in N} z_i(e) + \sum_{e \in S} \U(e)
        \geq \sum_{e \in E} \min\left\{\sum_{i \in N} z_i(e), \U(e)\right\} .
    \end{equation*}
    The first inequality holds with equality if and only if \eqref{eq:span-S} holds.
    The second inequality holds with equality if and only if both \eqref{eq:sat-E-S} and \eqref{eq:no-overload-S} hold.
\end{proof}

\subsection{Oracle models}\label{sec:oracle}

Auction style algorithms (in contrast to direct methods as described in \cite{leme2020computing}) have to be evaluated differently when we discuss computational efficiency.
The overall running time is heavily dependent on the valuation functions and hence, only pseudo-polynomial in the number of items and buyers in the worst case.
However, determining an overdemanded set to increase prices on (or more generally, any single price-adjustment step) is a computational problem which we should analyze under two different aspects:
\begin{enumerate*}
    \item how much information does the auctioneer require from the buyers to perform a price update, and
    \item given this information, how fast can this update be computed?
\end{enumerate*}
In summary, we should analyze the time to perform just a single step of the auction and the communication cost incurred in such a step instead of the total running time of the auction.
This seems reasonable as the total running time heavily depends on the number of steps, i.e., on something the auctioneer cannot influence with the limited information that is available in a dynamic auction.

Most types of dynamic auctions involve an auctioneer who communicates item prices to the buyers, who are then asked for a bundle $z$ of items, which maximizes their utilities $v_i(z) - \pr{p}{z}$.
In \cite{blumrosen2010computational}, these so called \emph{demand oracles} are compared against other natural types of oracles that partially reveal the private information of the buyers.
It should be mentioned that even for (strong) gross substitutes valuations, the size of preferred bundles can be exponential in $|E|$ and there is no more compact way to encode these sets as shown by \citet{knuth1974asymptotic}.

In this section, we describe the oracles we use in our algorithm and we compare them to the oracles used in the literature.
As our oracle models only rely on the polymatroid properties of the minimal, respectively maximal, preferred bundles at given prices $p$, we will use $\B_i$ as a placeholder for $\DM_i(p)$ and $\dhat_i(p)$.

We first define some weights, that play a key role in our algorithm.
Given $\bz = (z_1, \ldots, z_n)$ where $z_i \in \B_i$ for $i \in N$, let us define the following weights for each pair $\{e, f\} \subseteq E$:
\[
    w_i(e, f) \coloneqq \max \{\alpha \in \Z \mid z_i - \alpha \chi_f + \alpha\chi_e \in \B_i\}.
\]
Their importance is due to the following connection to $i$-tight sets:
\begin{lemma}\label{lem:tight-sets-weights}
    For an \M-convex set $\B_i$, $i \in N$ with associated rank function $\rk_i$, $z_i \in \B_i$, and $e, f \in E$, we have
    \begin{equation}\label{eq:weight-SFM}
    w_i(e, f) = \min \{\rk_i(S) - z_i(S) \mid e \in S, f \notin S\}.
\end{equation}
Thus,
    $w_i(e, f) > 0$ if and only if $f \in \Tt_{\B_i}(e,z_i)$.
    Further, a set $S$ is $i$-tight if and only if $\Tt_{\B_i}(e,z_i) \subseteq S$ for all $e \in S$.
\end{lemma}
\begin{proof}
    Let $\alpha \coloneqq w_i(e, f)$ and let $\beta$ denote the minimum value on the right-hand side of \eqref{eq:weight-SFM}, and $S$ a minimizer.
    By definition, $z'_i = z_i - \alpha\chi_f + \alpha\chi_e\in \B_i$.
    We first show $\alpha \leq \beta$.
    This follows since $z'_i(S) \leq \rk_i(S)$, and $z'_i(S) = z_i(S) + \alpha$, and therefore $\beta = \rk_i(S) - z_i(S) \geq z_i'(S) - z_i(S) = z_i'(S) - (z_i'(S) - \alpha) = \alpha$.

    Let us now turn to $\alpha \geq \beta$.
    For a contradiction, assume that $z_i'' = z_i - \beta\chi_f + \beta\chi_e\notin \B_i$, that is, $z''_i(T) > \rk_i(T)$ for some $T \subseteq E$.
    Thus, $z''_i(T) > z_i(T)$; by definition, this means that $e \in T$, $f \notin T$.
    Hence, $T$ is in the scope of the minimization problem in \eqref{eq:weight-SFM}, giving $\rk_i(T) - z_i(T)\geq\beta$.
    But this means that $z''_i(T) = z_i(T) + \beta \leq \rk_i(T)$, leading to a contradiction.
    This completes the proof of \eqref{eq:weight-SFM}.

    From here, the second statement is immediate.
    For the final statement, since $\Tt_{\B_i}(e,z_i)$ is the inclusion-wise minimal tight set containing $e$, it is clearly contained in every $i$-tight set $S \subseteq E$ with $e \in S$.
    Conversely, if $\Tt_{\B_i}(e,z_i) \subseteq S$ is true for every $e \in S$, then $\bigcup_{e \in S} \Tt_{\B_i}(e,z_i) = S$, which is $i$-tight as the union of $i$-tight set is $i$-tight by submodularity of the polymatroid rank function $\rk_i$.
\end{proof}

In the market context, we can interpret the weights $w_i(e,f)$ as part of the information required by the auctioneer from buyer $i$.
Given the current price vector $p$, each buyer $i$ is asked to report one minimal/maximal preferred bundle $z_i \in \B_i(p)$.
We denote the time of such a `demand oracle query' as $\BO$.
Further, we may query values $w_i(e, f)$, which is the maximum number $\alpha \in \Z$ such that buyer $i$ is willing to exchange $\alpha$ units of item $f$ against $\alpha$ units of item $e$.

Our algorithm needs access to these values $w_i(e, f)$; we denote the time required for computing one such value as $\EO$.
Their computational complexity for the buyers depends on the particular representation of the $\B_i$ sets.

A common model to represent the valuation functions $v_i$ is by a value oracle.
Given a value oracle, a minimal/maximal preferred bundle $z_i \in \B_i(p)$ can be computed using the greedy algorithm by $m^2$ value oracle calls.
Given $z_i \in \B_i$ and $\alpha \in \Z$, checking if $z_i - \alpha\chi_f + \alpha\chi_e \in \B_i(p)$ takes a single value oracle call.
Thus, one can compute $w_i(e, f)$ in $\mathcal{O}(\log(B))$ calls to the value oracle.
We can also formulate the computation of $w_i(e, f)$ as the submodular function minimization problem in \eqref{eq:weight-SFM}.
Hence, one can use a strongly polynomial submodular function minimization algorithm.
Nevertheless, this requires access to $\rk_i(S)$, which may lead to running a greedy algorithm for each call.
In \Cref{app:oracle_calls}, we provide an example that shows that the particular running time of $\BO$ and $\EO$ heavily depend on the representation of the valuation function, using the example of OXS functions.

For the unit-supply setting, the sets $\DM_i(p)$ and $\dhat_i(p)$ are matroid base sets and our task is to solve the matroid union problem.
For matroid optimization, the usual oracle settings are via independence or rank oracle queries.
Such oracles may be expensive if our primary oracle is for the valuation function $v_i$:
each independence or rank oracle query may require running a greedy algorithm.

To take the actual computation time to answer the oracle in some underlying structure into account, \citet{blikstad2023fast} propose to use the \emph{dynamic rank oracle}.
The idea is that given an independent set $S$, it is easy to find an independent set close to $S$, but much harder to find an independent set which is farther away (with respect to the symmetric difference of the sets).
To be precise, a dynamic rank oracle starts with the empty set $S$, then three different operations are denoted as oracle calls \begin{enumerate*} \item adding an element to $S$, \item deleting an element from $S$, and \item obtaining the rank of $S$
\end{enumerate*} (see also \cite[Definition 1.2]{blikstad2023fast}).

Our exchange oracle $\EO$ is efficient from this point of view, since given an independent set, we always query sets with symmetric difference two.
We can perform an exchange oracle query by using three dynamic rank oracles calls.\footnote{If we cannot perform the queried exchange, one can also argue that we need two dynamic oracle calls more to return to the original set, i.e., five dynamic oracle calls for one exchange oracle query.}
Note that we assume that we do not have to build the first independent set $z_i$ with the dynamic rank oracle but that we can perform the queries starting from the set that we obtained from the demand oracle.
A demand oracle query ($\BO$) can be answered by $\mathcal{O}(m)$ dynamic oracle queries.
The other way around, the dynamic rank oracle is stronger than the exchange oracle, in the sense that the dynamic rank oracle can compute the rank of a given set $S$ in $O(\abs{S})$ computations, while the exchange oracle cannot answer this question.
Note that in the algorithm for the unit-supply setting, we will use exchange operations that do not change the rank of the sets, so in fact we know the ranks as well.
It can therefore be said that the two models behave very similarly in the matroid case and our exchange oracle model can be generalized to polymatroids.

If we analyze our matroid union algorithm in the unit-supply setting with the dynamic rank oracle, we get a running time of $\mathcal{O}((m^3+nm^2) \cdot \DRO)$.
Hence, the running time of our matroid union algorithm in the dynamic rank oracle model is not comparable to the running time $\mathcal{O}(\mathrm{poly}(n \log(m))m^{3/2} \cdot \DRO)$ by \citet{blikstad2023fast} as we are linear in $n$, while \citet{blikstad2023fast} achieve a better dependency on $m$.

\section{Finding the overdemanded and underdemanded sets}
\label{sec:finding_over_and_underdemanded_sets}
As described earlier, minimal maximal overdemanded and minimal maximal underdemanded sets play an important role in the context of dynamic auctions.
In this section, we discuss their connection to the corresponding polymatroids and their rank functions.
Remember that $\DM_i(p)$ and $\dhat_i(p)$ are the base sets of integral polymatroids with associated rank functions $\rkmin$ and $\rkmax$.
These rank functions will help us to formally express over- and underdemandedness.

Moreover, we will see in \Cref{sec:exchangegraph} and \Cref{sec:exchange_graph_underd} how the minimal maximal over-/underdemanded sets can be computed using an optimal solution to the corresponding polymatroid sum problem.

\subsection{Overdemanded sets}

\label{sec:finding-min-max-overdemanded-set}
First, we consider minimal preferred bundles and their connection to maximal overdemanded sets.
For any price vector $p \in \R^E_+$, any set $S \subseteq E$, and any buyer $i \in N$, we define
\begin{equation*}
    \rkmin_i^p(S) \coloneqq \max\{z(S) \mid z \in \DM_i(p)\} \hspace{5ex} \mbox{ and } \hspace{5ex}
    \check{\mr}_i^p(S) \coloneqq \min\{z(S) \mid z \in \DM_i(p)\}
\end{equation*}
as the maximum, respectively minimum, number of items from $S$ in some minimal preferred bundle.
By \Cref{lemma:rk-demand_sets}, the set $\DM_i(p)$ is implicitly given by $\rkmin_i^p$, and also $\check{\mr}_i^p$ can be expressed in terms of $\rkmin_i^p$ via $\check{\mr}_i^p(S) = \rkmin_i^p(E) - \rkmin_i^p(E \setminus S).$
Further, for $p \in \R^E_+$ and $S \subseteq E$, we define the \emph{overdemandedness} of $S$ at prices $p$ by
\[
    \overd^p(S) \coloneqq \sum_{i \in N} \check{\mr}_i^p(S) - \U(S) \  = \sum_{i \in N} (\rkmin_i^p(E) - \rkmin_i^p(E \setminus S)) - b(S).
\]

\begin{definition}\label{def:overdemanded}
    We call a set $S \subseteq E$ \emph{overdemanded} w.r.t.\ prices $p$ if $\overd^p(S) > 0$.
    Moreover, a set $S \subseteq E$ is called \emph{maximal overdemanded} if it is overdemanded and maximizes the overdemandedness $\overd^p(S)$ over all bundles $S \subseteq E$.
\end{definition}
We may omit the superscript $p$ in $\overd^p(S)$ when it is clear from the context.
The definition is quite intuitive since $\check{\mr}_i^p(S)$ denotes the least number of items buyer $i$ likes to receive from $S$ and, thus, $\sum_{i \in N} \check{\mr}_i^p(S)$ can be seen as the demand, whereas $\U(S)$ denotes the supply.
Note that in the presence of an overdemanded set, the prices cannot be packing.
In \Cref{lem:packing}, we show that the converse is also true, namely, a price vector is packing if and only if no overdemanded set exists.

\begin{lemma}\label{lem:packing}
    Given a market instance with strong gross substitutes valuations, prices $p$ are packing if and only if there is no overdemanded set.
\end{lemma}
In the unit-supply setting, this corollary was already shown by \citet{gul2000english}.

\begin{proof}
    Let $p$ be a packing price vector and $\bz=(z_1,\dots,z_n)$ the corresponding packing allocation.
	Hence $z_i \in \DM_i(p)$, i.e., $z_i(E) = \rkmin_i(E)$ holds.
	This yields
    \[
        \sum_{i \in N} \rkmin_i(E)
        = \sum_{i \in N} z_i(E)
        = \sum_{e \in E} \sum_{i \in N} z_i(e)
        = \sum_{e \in E} \min \left\{\sum_{i \in N} z_i(e), \U(e)\right\}
    \]
    where the last equality holds as $\bz$ is packing ($\sum_{i \in N} z_i(e) \leq b(e)$ for all $e \in E$).
 Hence, \Cref{thm:poly-union} allows us to conclude that
    \begin{align}
        \sum_{i \in N} \rkmin_i(E) &= \sum_{e \in E} \min \left\{\sum_{i \in N} z_i(e), \U(e)\right\} \leq \max_{\bz'\in\DM^\times} \sum_{e \in E} \min \left\{\sum_{i \in N} z'_i(e), \U(e)\right\}
        \notag\\
        &= \min_{S \subseteq E} \left\{\sum_{i \in N} \rkmin_i(E \setminus S) + \U(S)\right\}.\label{eq:rank-min}
    \end{align}
    The outer inequality can be transformed into the following formula, which can be seen as a \emph{generalized Hall formula}:
    \begin{equation}\label{eq:Hall}
        \sum_{i \in N} \left(\rkmin_i(E) - \rkmin_i(E \setminus S)\right) \leq \U(S) \quad \text{for all } S \subseteq E.
    \end{equation}
    This is again equivalent to the fact that there is no overdemanded set $S$.
    \medskip
    To show the reverse direction, assume that there is no overdemanded set at prices $p$.
    Thus, we apply equation~\eqref{eq:Hall}, which yields
    \begin{align*}
        \sum_{i \in N} \rkmin_i(E)
        &\leq \min_{S \subseteq E} \left\{\sum_{i \in N} \rkmin_i(E \setminus S) + \U(S)\right\}
        = \max_{\bz \in \DM^\times} \left\{\sum_{e \in E} \min \left\{\sum_{i \in N} z_i(e), \U(e)\right\}\right\} \\
        &\leq \max_{\bz \in \DM^\times} \left\{\sum_{e \in E} \sum_{i \in N} z_i(e)\right\}
        = \max_{\bz \in \DM^\times} \left\{\sum_{i \in N} z_i(E)\right\}
        =\sum_{i \in N} \rkmin_i(E) .
    \end{align*}
    Hence, equality must hold throughout the equation.
	In particular, this means that there exists an allocation $\bz \in \DM^\times$ with $\sum_{i \in N} z_i(e) \leq b(e)$ for all $e \in E$, i.e., the allocation $\bz$ is packing and the price vector is packing as well.
\end{proof}

Next, we will consider the connection of maximal overdemanded sets to an optimal solution $\bz$ of the polymatroid sum problem
\begin{equation}\label{eq:max_dcheck}
    \max_{\bz \in \DM^\times} \sum_{e \in E} \min \left\{\sum_{i \in N} z_i(e), \U(e)\right\}.
\end{equation}

This is done in several steps.
First, we characterize which properties a minimal maximal overdemanded set needs to satisfy and introduce a technical lemma.
Then, we combine these two statements to describe a nice way to compute an actual minimal maximal overdemanded set by constructing an auxiliary digraph and choosing all items from which we can reach an oversold item.
We say that an item $e \in E$ is \emph{undersold} if $\sum_{i \in N} z_i(e) < \U(e)$, and \emph{oversold} if $\sum_{i \in N} z_i(e) > \U(e)$.
For $i \in N$ and $z_i \in \DM_i$, let us call a set $S \subseteq E$ \emph{$i$-tight} if $z_i(S) = \rkmin_i(S)$.
To shorten notation slightly, we write, for $e \in E$,
\[
    \Tt_i(e) \coloneqq \Tt_{\DM_i}(e, z_i)
\]
for the minimal $i$-tight set containing $e$ when $z_i$ and $\DM_i$ are clear from the context.

\begin{lemma}\label{lem:minmax-over}
    Assume $\bz = (z_1, \ldots, z_n)$ is an optimal solution to \eqref{eq:max_dcheck}.
    Then, a set $S$ is maximal overdemanded if and only if it fulfills the following properties:
    \begin{itemize}
        \item $E\setminus S$ is $i$-tight for every $i \in N$;
        \item $S$ does not include any undersold items;
        \item $S$ includes all oversold items.
    \end{itemize}
   Moreover, there is a unique inclusion-wise minimal set fulfilling these properties.
\end{lemma}
\begin{proof}
    Let $S\subseteq E$.
    We can write
    \begin{align*}
        \overd(S) = \sum_{i \in N} \rkmin_i(E) - \left(\sum_{i \in N} \rkmin_i(E \setminus S) + \U(S)\right).
    \end{align*}
    Hence, $S$ is a maximal overdemanded set if and only if it is a minimizer of the right-hand side of~\eqref{eq:max-min}.
    Since $\bz$ is optimal, we get, by \Cref{thm:poly-union},
    \[
        \sum_{e \in E} \min \left\{\sum_{i \in N} z_i(e), \U(e)\right\} = \sum_{i \in N} \rkmin_i(E \setminus S) + \U(S).
    \]
    By Lemma~\ref{lem:poly-union-opt}, this is the case if and only if $S$ satisfies all three properties stated in the lemma.
    Moreover, if $S$ and $S'$ satisfy these three properties, then $S \cap S'$ and $S \cup S'$ also satisfy them.
    Hence, there exists a unique inclusion-wise minimal set satisfying these three properties.
\end{proof}

In \Cref{sec:exchangegraph} we will see that we can compute a minimal maximal overdemanded set by performing a breadth-first search in the exchange graph corresponding to an optimal solution $\bz$ of the polymatroid sum problem~\eqref{eq:max_dcheck}.

\subsection{Finding minimal maximal overdemanded sets via the exchange graph}\label{sec:exchangegraph}
We define the exchange graph $\check{G}_i(z_i) = (E, \check{A}_i)$ for $i \in N$ by setting
\[
    \check{A}_i \coloneqq \{(e, f) \in E \times E \mid \check{w}_i(e, f) > 0\},
\]
where $\check{w}_i(e, f) \coloneqq \max \{\alpha \in \Z \mid z_i - \alpha \chi_f + \alpha\chi_e \in \DM_i\}$ is the weight which can be computed by an exchange oracle call (ExO) as described in \Cref{sec:oracle}.
We let $\check{G}(\bz) = (E, \check{A})$, where $\check{A} \coloneqq \bigcup_{i \in N} \check{A}_i$.

\begin{theorem}\label{thm:computation_overd-set}
    Assume $\bz$ is an optimal solution to \eqref{eq:max_dcheck}.
    Then the minimal maximal overdemanded set is the set $R$ of items from which we can reach an oversold item in $\check{G}(\bz)$.
\end{theorem}

\begin{proof}
    Let $\bz$ be an optimal solution to \eqref{eq:max_dcheck} and $R$ the set of items from which an oversold item can be reached in $\check{G}(\bz)$.
    We will show that $R$ is a maximal overdemanded set and, moreover, included in any other maximal overdemanded set.

    To this end, let $S$ be a maximal overdemanded set.
	Then by \Cref{lem:minmax-over}, $S$ contains all oversold items and no undersold item.
    Moreover, $E \setminus S$ is tight and thus, by \Cref{lem:tight-sets-weights}, there is no arc leaving $E \setminus S$, i.e., $\delta^-(S) = \emptyset$.
    \begin{figure}
    \centering
        \begin{tikzpicture}[scale=0.24, transform shape, every node/.style={scale = 3.8}]
            \draw[rwthblue, fill=rwthblue!10] (0,0) ellipse (5 and 10);
            \node[draw=none, rwthblue] at (-3,10) {\small $S$};
			\draw[rwthblue, fill=rwthblue!10] (14,0) ellipse (5 and 10);
            \node[draw=none, rwthblue] at (11,10) {\small $E\setminus S$};
            \draw[rwthgreen, fill=rwthgreen!10] (0,2) ellipse (4.5 and 6);
            \draw[rwthorange, fill=rwthorange!10] (0,4) ellipse (3.2 and 3.2);
            \node[rwthorange, fill=none, draw=none] (o) at (0,5) {\small oversold};
            \draw[rwthorange, fill=rwthorange!10] (14,4) ellipse (3.2 and 3.2);
            \node[rwthorange, fill=none, draw=none] (o) at (14,5) {\small undersold};
            \draw [->, rwthred, thick] (11,-5) --node[midway]{$\mathbin{\tikz [x=1.4ex,y=1.4ex,line width=.2ex, rwthred,-] \draw (0,0) -- (1,1) (0,1) -- (1,0);}$} (3,-5);

		      \node[item] (1) at (0,2) {};
		      \node[item] (2) at (-1.5,0) {};
		      \node[item] (3) at (-1.5,-2) {};
		      \node[item] (4) at (1,-2) {};
		      \node[item] (5) at (1.5,0.5) {};
		      \node[item] (6) at (-1.5,3) {};
		      \node[item] (7) at (-3.25,2) {};

		      \draw[->] (2) -- (1);
		      \draw[->] (3) -- (2);
		      \draw[->] (4) -- (2);
		      \draw[->] (5) -- (1);
		      \draw[->] (7) -- (6);

		      \node[draw=none, rwthgreen] at (2.5,-1) {\small $R$};

            \draw [->, rwthred, thick] (-3,-5) --node[midway, rotate=52]{$\mathbin{\tikz [x=1.4ex,y=1.4ex,line width=.2ex, rwthred,-] \draw (0,0) -- (1,1) (0,1) -- (1,0);}$} (3);
	\end{tikzpicture}
    \caption{A maximal overdemanded set $S$ and the set of items $R$ from which an oversold item can be reached.}
    \label{fig:proof_overdemanded_set}
    \end{figure}
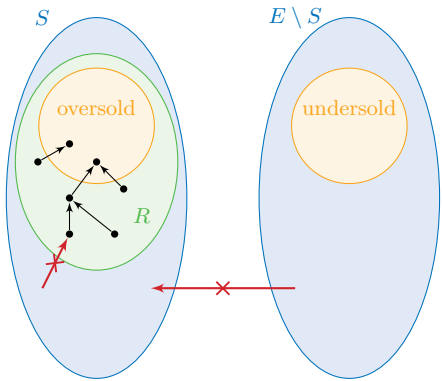
    Now, consider the set $R$.
    A sketch of $S$ and $R$ can be found in \Cref{fig:proof_overdemanded_set}.
    By definition, $R$ contains every oversold item.
    Moreover, as all undersold items are in $E\setminus S$, all oversold items are in $S$, and $\delta^-(S) = \emptyset$, there is no path from an undersold item to an oversold item.
    Hence, the undersold items are in $E \setminus R$.
	Furthermore, by definition of $R$, there is no arc entering $R$, i.e., the set $E \setminus R$ is tight.
    Hence $R$ is maximal overdemanded.

    Finally, observe that $R \subseteq S$, as otherwise, there would be an arc entering $S$ which yields a contradiction.
	This means that the set $R$ is included in every maximal overdemanded set and is thus the inclusion-wise minimal maximal overdemanded set.
\end{proof}

\begin{lemma}\label{lem:running_time_overd_set}
    Overdemanded sets can be computed in time $\mathcal{O}(nm^2 \cdot\EO)$, given an optimal solution $\bz$ to the polymatroid sum problem \eqref{eq:max_dcheck}.
\end{lemma}
\begin{proof}
    First, we construct the exchange graph.
    We consider each of the $m$ items $e \in E$ once and do an edge-weight query for each of the $n$ buyers $i \in N$ and her currently assigned items.
    Thus, we have a running time of $\mathcal{O}(nm^2\cdot\EO)$ for this step.
    Then we do a breadth-first search in the graph, starting from all oversold items in reverse direction of the arcs.
    This takes time $\mathcal{O}(m + m^2)$.
    Hence, we need $\mathcal{O}(nm^2\EO + m^2)$ time to compute the set of items from which we can reach any oversold item, which is, by \Cref{thm:computation_overd-set}, equal to the minimal maximal overdemanded set.
\end{proof}

In \Cref{sec:push-relabel}, we will present an algorithm to compute an optimal solution $\bz$ of the polymatroid sum problem \eqref{eq:poly-union}.
This will lead to a total running time of $\mathcal{O}(n \cdot \BO + nm^3 \cdot \EO)$ to compute an overdemanded set in the multi-supply setting and to an improved running time of $\mathcal{O}(n\cdot\BO+(m^3 + nm^2)\cdot\EO)$ in the unit-supply setting.

\subsection{Underdemanded sets}
The notion of underdemandedness is complementary to the notion of overdemandedness.
Intuitively, it measures how many items from a set $S$ will remain unsold, even if every buyer chooses a preferred bundle with the maximum number of items from $S$.

Here, we can make use of the rank function of the polymatroid base polytope given by the maximal demand sets $\dhat_i(p)$, i.e, by
\[\rkmax^p_i(S)\coloneqq \max\{z(S)\mid z\in \dhat_i(p)\}.\]

This allows us to express the underdemandedness of a set $S$ for prices $p$ as
\[\underd^p(S) \coloneqq b(S)-\sum_{i\in N} \rkmax_i^p(S).\]
\begin{definition}\label{def:underdemanded}
	We call a set $S \subseteq E$ \emph{underdemanded} w.r.t.\ prices $p$ if $\underd^p(S) > 0$.
	Moreover, a set $S \subseteq E$ is called \emph{maximal underdemanded} if it is underdemanded and maximizes the underdemandedness $\underd^p(S)$ over all bundles $S \subseteq E$.
\end{definition}
In the following, we will omit the superscript $p$ to improve readability.
\begin{restatable}{lemma}{LemCovering}\label{lem:covering}
    Given a market instance with strong gross substitutes valuations, prices $p$ are covering if and only if there is no underdemanded set.
\end{restatable}

The lemma follows, similarly to \Cref{lem:packing}, from the Min-Max Theorem for polymatroid sum (\Cref{thm:poly-union}).
The proof can be found in \Cref{app:underdemanded_sets}.
There, we rephrase
\[\underd(S)=b(S)-\sum_{i\in N} \rkmax_i(S)=b(E)-\left(b(E\setminus S)+\sum_{i\in N} \rkmax_i(E\setminus (E\setminus S))\right).\]
By this, we obtain the following corollary.
\begin{corollary}\label{lemma:max_underdemanded_iff}
	$S\subseteq E$ is a maximal underdemanded set if and only if
	\[E\setminus S\in \arg\min \left\{b(T)+\sum_{i\in N} \rkmax_i(E\setminus T) \ \Big\vert \ T\subseteq E \right\}.\]
    Moreover, $S$ is a minimal maximal underdemanded set if and only if $E\setminus S$ is a maximal minimizer of $T\mapsto b(T)+\sum_{i\in N} \rkmax_i(E\setminus T)$.
\end{corollary}

In order to determine a minimal maximal underdemanded set $S$, we follow the same approach as in \Cref{sec:finding-min-max-overdemanded-set} using the Min-Max Theorem for polymatroid sum (\Cref{thm:poly-union}).
Given an optimum solution $\bz=(z_1,\ldots,z_n)$ to the polymatroid sum problem
\begin{equation}
    \label{eq:max_dhat}
    \max_{\bz \in \dhat^\times} \left\{\sum_{e \in E} \min \left\{\sum_{i \in N}z_i(e), b(e)\right\}\right\},
\end{equation}
we explain how to find a minimal maximal underdemanded set $S$.
Recall that for $S\subseteq E$ and $i\in N$, we call $S$ \emph{$i$-tight} if $z_i(S)=\rkmax_i^p(S)$.
We obtain the following analogue of \Cref{lem:minmax-over}:
\begin{restatable}{lemma}{LemMinMaxUnder}
    \label{lem:minmax-under}
	Assume $\bz = (z_1, \ldots, z_n)$ is an optimal solution to \eqref{eq:max_dhat}.
	Then, a set $S$ is maximal underdemanded if and only if it fulfills the following properties:
	\begin{itemize}
		\item $S$ is $i$-tight for every $i \in N$;
		\item $S$ includes all undersold items;
		\item $S$ does not include any oversold items.
	\end{itemize}
    Moreover, there is a unique inclusion-wise minimal set fulfilling this properties.
\end{restatable}
It can be derived from \Cref{lem:poly-union-opt}, similar to the proof of \Cref{lem:minmax-over} (see \Cref{app:underdemanded_sets}).

\subsection{Finding minimal maximal underdemanded sets}
\label{sec:exchange_graph_underd}
Similar to \Cref{sec:exchangegraph}, we can compute the minimal maximal underdemanded set $S$ via reachability in an exchange graph.
Given prices $p$, $z_i \in \dhat_i(p)$ and $e,f\in E$, we define weights, which can be computed by an exchange oracle as described in \Cref{sec:oracle}, by
\[\widehat{w}^i_i(e,f)\coloneqq \max\{\alpha\in \Z\mid z_i-\alpha \chi_f +\alpha\chi_e \in \dhat_i(p)\}.\]
Let further
\[
\widehat{A}_i \coloneqq \{(e, f) \in E \times E \mid \widehat{w}^p_i(e, f) > 0\} \quad \text{ for $i\in N$},
\]
and define $\widehat{G}(\bz) \coloneqq (E, \widehat{A})$ with $\widehat{A} \coloneqq \bigcup_{i \in N} \widehat{A}_i$.

\begin{restatable}{theorem}{ThmCompUnderdSet}
    \label{thm:computation_underd-set}
	The minimal maximal underdemanded set $S$ is the set $R$ of items which are reachable from an undersold item in $\widehat{G}(\bz)$.
\end{restatable}
This proof follows from \Cref{lem:tight-sets-weights} and \Cref{lem:minmax-under}.
It is similar to the one of \Cref{thm:computation_overd-set} and given in \Cref{app:underdemanded_sets}.

In particular, we obtain the following analogue of \Cref{lem:running_time_overd_set}:
\begin{lemma}\label{lem:running_time_underd_set}
	The minimal maximal underdemanded set can be computed in time $\mathcal{O}(nm^2 \cdot\EO)$, given an optimal solution~$\bz$
    to \eqref{eq:max_dhat}.
\end{lemma}

In the following section, we present an algorithm to compute an optimal solution $\bz$ to the general polymatroid sum problem \eqref{eq:poly-union}.
In particular, this algorithm allows us to compute optimal solutions to \eqref{eq:max_dcheck} or to \eqref{eq:max_dhat}, respectively.

\section{Push-relabel algorithm to solve the polymatroid sum problem}\label{sec:push-relabel}
Push-relabel algorithms are a well-studied algorithmic paradigm to solve submodular function minimization problems or subclasses of those.
We  start by presenting a general push-relabel framework for polymatroids and prove its correctness.
Afterwards, we show how to implement this framework first for matroids and then for polymatroids and analyze the running time.
Note that, in principle, any algorithm could be used to solve the matroid union (polymatroid sum) problem in order to compute the minimal maximal over-/underdemanded set with the help of the exchange graph as described in Section \ref{sec:exchangegraph} above.
In this section, we present an efficient implementation of a push-relabel algorithm for matroid union and polymatroid sum with respect to the number of $\EO$ oracle queries.
Our push-relabel algorithm is based on \citet{frank2012simple} who present a simpler algorithm for matroid union and a more general one for submodular flow feasibility; our algorithm generalizes the first and is a special case of the second.
In contrast to our analysis, \citet{frank2012simple} only give a complexity bound in terms of `basic operations', i.e., push-relabel steps.
We give a self-contained presentation of the algorithm and the analysis.

\subsection{The push-relabel framework}\label{sec:push-relabel-framework}
Let $\B_1,\dots,\B_n \in \rng$ be M-convex sets.
Further, we consider the weights $w_i(e,f)$ with respect to an $i \in N$ and a fixed $z_i \in \B_i$ as described in \Cref{sec:oracle}.
We define the level function $\level\colon E \to \{0, 1, \ldots, m\}$, and denote $\lmin(S) \coloneqq \min \{\level(e) : e \in S\}$ for $S \subseteq E$.
Recall that $E = \{1, 2, \ldots, m\}$, i.e., items are labeled by integers and the same holds for buyers.
Thus, we may use the usual order $\leq$ on each of them.
Further recall that we use the shorthand notation $\Tt_i(e)$ instead of $\Tt_{\B_i}(e, z_i)$ for the unique minimal $i$-tight set containing $e$.

\begin{definition}
    Given $z_i \in \B_i$ for all $i \in N$, a mapping $\level\colon E \to \{0, 1, \ldots, m\}$ is a \emph{valid level function} if it satisfies the following properties:
    \begin{enumerate}[label=(L\arabic*),leftmargin=*]
        \item\label{i:l1} $\level(e) = 0$ whenever $e \in E$ is oversold,
        \item\label{i:l2} $\lmin(\Tt_i(e)) \geq \Theta(e) - 1$ for all $i \in N$ and $e \in E$.
    \end{enumerate}
\end{definition}

The general idea of the push-relabel algorithm is as follows.
Initially, all $e \in E$ have $\level(e) = 0$.
As long as there still is an undersold item with label below $m$, we pick one such item $e$.
Now, we aim to push at this item if possible or relabel otherwise, where the push-relabel operations are described as follows:

\begin{description}
    \item[Push:] If we find an item $f$ with $f \in \Tt_i(e)$ and $\level(f) = \level(e) - 1$, then we perform a \emph{push} at $e$ with respect to $i$ and $f$.
        In such a push, we replace $z_i$ by $z_i' \coloneqq z_i - \alpha\chi_f + \alpha\chi_e$ for ${\alpha \coloneqq \min \{\U(e) - \sum_{\ell \in N} z_\ell(e), w_i(e, f)\}}$.
        That is, we select the largest value of $\alpha$ such that $e$ does not get oversold, and such that $z'_i\in \B_i$.
    \item[Relabel:] If no such item exists, we \emph{relabel} by setting $\level(e) \coloneqq \level(e)+1$.
\end{description}

If in the push operation, the minimum in $\alpha \coloneqq \min \{\U(e) - \sum_{\ell \in N} z_\ell(e), w_i(e, f)\}$ is attained in the second entry, this is called a \emph{saturating push at $e$ w.r.t. $f$ and $i$};
otherwise, if the minimum is attained only in the first entry, it is called a \emph{non-saturating push at $e$ w.r.t. $f$ and $i$}.

\begin{algorithm}
    \SetAlgoRefName{Push-Relabel Algorithm}
    $z_i \coloneqq \text{any vector from } \B_i$ for all $i \in N$\\
    $\level(e) \coloneqq 0$ for all $e \in E$\\
    \While{there is an undersold item $e$ with $\level(e) < m$}{
        choose an undersold item $e$ with $\level(e) < m$\\
        \quad \textbf{Push} if there is some $f$ with $f \in \Tt_i(e)$ and $\level(f) = \level(e) - 1$\\
        \quad \textbf{Relabel} otherwise
    }
    \Return{$\bz$}
    \caption{}
    \label{alg:polymatroid_union}
\end{algorithm}

The following two lemmata state that the level invariants are maintained, and that the push-relabel framework indeed computes an optimal solution.

\begin{lemma}\label{lem:level_invariants_and_stopping}
    The level invariants \ref{i:l1} and \ref{i:l2} are maintained throughout the \ref{alg:polymatroid_union}.
\end{lemma}

\begin{proof}
    The invariants \ref{i:l1} and \ref{i:l2} clearly hold at initialization.
    To see that \ref{i:l1} is maintained, note that labels are only modified for undersold items, and no undersold item gets oversold during the execution of the algorithm (by the choice of $\alpha$ in every push operation).
    Note that it may happen that an initially oversold item $e$ with $\level(e) = 0$ gets undersold, but then it never gets oversold again.

    Let us now turn to \ref{i:l2}.
    First, consider a relabel operation of some item $e$.
    Since no push was performed, it follows that for all $f \in \Tt_i(e)$, it holds that $\level(f) > \level(e)-1$ (smaller is not possible, since \ref{i:l2} is fulfilled before the relabeling operation).
    Thus, $\lmin(\Tt_i(e)) \geq \level(e)$, and hence, \ref{i:l2} is still fulfilled after the level of $e$ is increased by one.

    Consider now a push operation that changes $z_i$ to $z'_i\coloneqq z_i - \alpha\chi_f + \alpha\chi_e$ for $\alpha > 0$.
    Let $\Tt'_j(g)$ denote the $j$-tight sets for each $g \in E$ after the push operation.
    We need to show
    \begin{equation}\label{eq:Tt-prime}
        \lmin(\Tt'_j(g)) \geq \level(g) - 1 \quad \text{for all } j \in N, g \in E.
    \end{equation}
    This inequality is immediate for all $j \neq i$ and $g \in E$ since the bundle $z_j$ and the levels do not change, therefore $\Tt'_j(g) = \Tt_j(g)$.
    It also remains true for $j = i$, and for all $g \in E$ such that $f \notin \Tt_i(g)$, since in this case $\Tt_i(g)$ is also $i$-tight for $z'_i$, and thus $\Tt_i'(g) \subseteq \Tt_i(g)$.

    Hence, it is left to show \eqref{eq:Tt-prime} for the case $j = i$, and $f \in \Tt_i(g)$.
    We  have $\Tt_i'(g) \subseteq \Tt_i(g) \cup \Tt_i(e)$, since $\Tt_i(g) \cup \Tt_i(e)$ is tight before the push operation as the union of two tight sets, and remains tight.
    Thus,
    \[
        \lmin(\Tt'_i(g)) \geq \min \left\{\lmin(\Tt_i(g)), \lmin(\Tt_i(e))\right\} = \min \left \{\lmin(\Tt_i(g)), \level(f)\right\} = \lmin(\Tt_i(g)),
    \]
    where the penultimate equality follows since $\level(f) = \lmin(\Tt_i(e))$ if $f$ is selected for a push (by the selection property and since \ref{i:l2} is fulfilled before the push).
    The last equality follows by $f \in \Tt_i(g)$.
    This shows that \eqref{eq:Tt-prime} remains true after the push operation.
\end{proof}

\begin{lemma}\label{lem:terminate}
    The \ref{alg:polymatroid_union} returns an optimal solution to the polymatroid sum problem~\eqref{eq:poly-union}.
\end{lemma}
\begin{proof}
    First, we observe that when the algorithm stops, there is a value $\ell \in \{0, 1, \ldots, m\}$ such that there is no item $e$ with $\level(e) = \ell$.
    To see this, note that we have only $m$ items but $m+1$ levels, so by pigeonhole principle, there exists an empty level $\ell \leq m$.

    By \Cref{lem:poly-union-opt},
    equation \eqref{eq:poly-union-leq} always holds.
    Using the equality criteria of \Cref{lem:poly-union-opt}, we show that the equation is fulfilled with equality for our computed bundle $\bz$ and $S \coloneqq \{e \in E \mid \level(e) < \ell\}$, implying that $\bz$ is an optimal solution to the polymatroid sum problem \eqref{eq:poly-union}.
    We note that $S = \emptyset$ and $S = E$ are both possible.

    Since there is no undersold item on any level below $m$ by the stopping criterion, there is also no undersold item on any level below $\ell$ and the equality criterion \eqref{eq:sat-E-S} holds by definition of $S$.
    Optimality criterion \eqref{eq:no-overload-S} is satisfied because all items in $E \setminus S$ have a label greater than 0 while all oversold items have label 0 by \ref{i:l1}.
    To show \eqref{eq:span-S}, note that $\Tt_i(e) \subseteq E \setminus S$ for all $e \in E \setminus S$ by \ref{i:l2} and the choice of $\ell$.
    Hence, $E \setminus S = \bigcup_{e \in E \setminus S} \Tt_i(e)$ is $i$-tight for each $i \in N$, as it is the union of $i$-tight sets.
\end{proof}

We have seen that the \ref{alg:polymatroid_union} returns an optimal solution $\bz$ to the polymatroid sum problem~\eqref{eq:poly-union}.
In the next two sections, we explain how to find the items to push or relabel in a fast way.
Therefore, we distinguish the unit-supply setting, i.e., where $\U(e) = 1$ for each $e \in E$ and the general multi-supply setting.

\subsection{Running time of the unit-supply case}\label{sec:push-relabel-matroid}
We now consider the unit-supply case, i.e., $\U(e) = 1$ for all $e \in E$.
In this case, each $ \B_i$ is an \M-convex subset of $[\bO, \mathbf{1}]_{\Z}$, which precisely corresponds to the set of bases of a matroid.
Hence, $z_i \in  \B_i$ is the indicator vector of a base; we will also use the notation $B_i = \supp(z_i)$.
In this case, it is easy to see that $\Tt_i(e) = \{e\}$ if $e \in B_i$, and $\Tt_i(e)$ is the fundamental circuit $C(B_i, e)$ in the $i$-th matroid if $e \in E \setminus B_i$.

For the special case of unit-supply, the \ref{alg:polymatroid_union} can be described as follows:
In each iteration with current bases $B_i$ for $i \in N$, and levels $\Theta\colon E \to \{0, \ldots, m\}$, the algorithm checks whether there exists an undersold, thus unsold, item $e \in E$ of level $\Theta(e) < m$, and if so, chooses one such item.

In the body of the \textbf{while}-loop, given an unsold item $e$, the algorithm checks if it can perform a push operation.
For this, the algorithm might need to go through all buyers $i \in N$, and check whether there exists an item $f$ in the fundamental circuit $C(B_i, e)$ of level $\Theta(f) = \Theta(e) - 1$.
This may cost $nm$ queries, but only if $e$ has level $\Theta(e) = 1$.
If $e$ has level $\Theta(e) = 0$, we can immediately relabel.
Whenever $e$ has level $\Theta(e) \geq 2$, then it suffices to go through each item $f$ at most once, since each item $f$ of level $\Theta(f) = \Theta(e) - 1$ that is sold is owned by exactly one buyer (by \ref{i:l1} and since there is only one unit per item).
The proof of the following theorem makes use of the described procedure to bound the running time.

\begin{theorem}\label{thm:running_time_polymatroid_union_unit_supply}
    The \ref{alg:polymatroid_union} can be implemented in running time $\mathcal{O}(n\cdot\BO+(m^3 + nm^2)\cdot\EO)$ for the unit-supply case.
\end{theorem}

\begin{proof}
    The algorithm starts with finding an initial allocation $\bz$ in $n\cdot\BO$ time.
    Given an allocation $\bz$, we can compute how many buyers own an item $e$ in $\mathcal{O}(nm)$.
    These values can be updated in time $\mathcal{O}(1)$ at every push operation.
    Moreover, these values allow us to maintain a list of all unsold items, which needs $\mathcal{O}(m)$ time for initialization and can be maintained in $\mathcal{O}(1)$ at each push.
    Note that at a relabel operation, nothing changes for these values and the list.

    We find an unsold item $e$ and check its level $\level(e)$ in $\mathcal{O}(1)$.
    If $\level(e) = 0$, we can relabel immediately.
    If $\level(e) = 1$, we will go through all buyer-item combinations of items in level $0$ to determine a push (if possible, else we relabel).
    This takes $\mathcal{O}(mn)$ exchange oracle calls.
    For every other level $\ell \geq 2,$ any item on level $\ell-1$ belongs to at most one buyer (by \ref{i:l1}, and since there is only one unit per item).
    Thus, we just need to go through the items once to determine a push or a relabel operation.
    This can be done in $\mathcal{O}(m)$ exchange oracle calls.

    The number of relabel operations per level $\ell \geq 1$ is bounded by $\mathcal{O}(m)$.
    The number of pushes per level is also bounded by $\mathcal{O}(m)$.
    To see this, we show that for each level $\ell$ the value $\Phi(\ell) = \sum_{i \in N} \abs{B_{i, \geq \ell}}$ is non-decreasing in every relabel or push operation, where $\abs{B_{i, \geq \ell}}$ is the number of items in the base of buyer $i$ that were added to the base when their level was at least $\ell$.
    Relabel operations are only used for unsold items, and hence they do not affect the $B_{i, \geq \ell}$ sets; thus, $\Phi(\ell)$ does not change.
    At any push at $e$ with respect to $f$ and $i$, $\abs{B_{j, \geq \ell}}$ does not change for $j \neq i$.
    For buyer $i$, the item $f$ which leaves the base $B_i$ is exactly one level lower than the the level of the item $e$ which enters the base.
    Hence, $\abs{B_{i, \geq \ell}}$ is not decreasing.
    When pushing at $e$ with respect to $i$ and $f$ with $\level(e) = \ell$, we strictly increase $\abs{B_{i, \geq \ell}}$ and thus, $\Phi(\ell)$.

    Since $\Phi(\ell) \leq m$ for each $\ell \geq 1$ (as each item with level $\ell\geq 1$ appears in just one base), we get a bound of $m$ on the number of pushes per level.

    This gives a total running time of $\mathcal{O}((m^2 n + m^3)\cdot\EO)$ after finding the initial allocation, where the first term is generated by the pushes and relabel operations on items at level 1 and the second term captures all remaining pushes.
\end{proof}

\begin{corollary}
    Given prices $p$, we can compute the minimal maximal overdemanded [underdemanded] set in time $\mathcal{O}(n \cdot \BO + (m^3 + nm^2) \cdot \EO)$ in the unit-supply setting.
\end{corollary}
This follows by \Cref{lem:running_time_overd_set} [\Cref{lem:running_time_underd_set}, respectively] and \Cref{thm:running_time_polymatroid_union_unit_supply}.
Note that the demand and exchange oracle calls are based on the matroids corresponding to the preferred bundles $\DM_i(p)$ or $\dhat_i(p)$, respectively.

\subsection{A polymatroid sum algorithm}
In this section, we provide a fast implementation of the push-relabel framework for the multi-supply setting.
We will need a different approach here than in the unit-supply setting.
There, the key property of the approach is that an item which is not oversold is owned by at most one buyer.
This is not true any longer.

Recall that the generic push-relabel algorithm selects an undersold item $e$ and then searches for a buyer-item pair $(i, f)$ satisfying $w_i(e, f) > 0$ and $\level(f) = \level(e)-1$.
If such a pair exists, a push is performed.
If not, item $e$ gets relabeled.
In this section, we show that it is possible to go through the buyer-item pairs in a structured way.
Namely, for a fixed item $e$ we consider the buyer-item pairs in lexicographically increasing order.
This (total) lexicographical ordering is based on the ordering of the buyers $N = \{1, \ldots, n\}$ and the items $E = \{1, \ldots, m\}$ by natural numbers, i.e., $(i, f)$ is lexicographically smaller than $(i', f')$ if $i < i'$, or if $i = i'$ and $f < f'$.

In \Cref{lem:lexicographic_selection}, we show a crucial monotonicity property: For a fixed item $e$, the lexicographically minimal buyer item pair $(i, f)$ which fulfills all necessary conditions for a push will lexicographically only increase until $e$ is relabeled.
This implies that for a fixed item $e$, we can find all possible pushes until $e$ is relabeled by going in lexicographically increasing order through all item-buyer combinations.

We note that lexicographic selection rules are commonly used in submodular optimization.
This idea traces back to the first polynomial-time algorithm for polymatroid intersection by \citet{schonsleben1980ganzzahlige}, and was used in the push-relabel algorithms by \citet{fujishige1992new,fujishige1996push}.
Notably, the algorithm by \citet{frank2012simple} does not rely on this rule.
It is beneficial in our implementation as it enables to decrease the number of exchange oracle queries, leading to an implementation of the push-relabel framework in $\mathcal{O}(m^3 n \cdot \EO)$.

\subsubsection{Algorithm.}
We implement the general push-relabel algorithm with two subtleties.
First of all, we always select an undersold item $e$ of highest level among those of level strictly smaller than $m$.
Second, for some given item $e$, when selecting the buyer-item pair $(i, f)$ for the next push, we always select the pair which is lexicographically smallest among all buyer-item pairs that fulfill the necessary requirements for a push at item $e$.

To implement this rule efficiently, we exploit the monotonicity property stated in \Cref{lem:lexicographic_selection}.
This means that for every item $e$, we keep track of the next buyer-item pair which we have to consider as a candidate for a push in case $e$ gets selected.
Starting from this pair, we search for a feasible candidate in lexicographically increasing order.
When a candidate is found, a push is performed and the pointer to the next buyer-item pair to consider is updated.
For a non-saturating push, this is the buyer-item pair for which the push was performed, for a saturating one it is the lexicographically next buyer-item pair.
If no feasible pair for a push can be found, this means that $e$ gets relabeled and the pointer to the next buyer-item pair to consider is reset to the lexicographically minimal one.
We give an implementation of the described procedure in pseudo code.

\begin{algorithm}[htbp]
    \DontPrintSemicolon
    \SetKw{Break}{break}
    \SetAlgoRefName{Polymatroid Sum Algorithm via Push-Relabel}
    $z_i \coloneqq$ any vector from $\B_i$ for all $i \in N$\\
    $\level(e) \coloneqq 0$ for all $e \in E$\\
    initialize pointer to the buyer-item pair for the exchange at item $e$ as $(j_e,g_e) \coloneqq (1,1)$ for all $e \in E$\\
    \While{there is an undersold item $e$ with $\level(e)<m$}{
        $push \coloneqq false$\\
        Choose $e \in \arg\max_{e \in E} \{\level (e) \mid e \text{ is undersold and } \level(e)<m\}$ \\
        \For{buyer-item pairs $(i, f)$ in lexicographically increasing order starting with the pair $(j_e, g_e)$}{
           \If(\tcp*[f]{buyer-item pair to perform a push}){$w_i(e, f) > 0$ \AAnd $\level(e) = \level(f) + 1$}{
                $push \coloneqq true$\\
                $\alpha \coloneqq \min \{ \U(e) - \sum_{\ell \in N} z_\ell(e), w_i(e,f)\}$\\
                $z_i \coloneqq z_i - \alpha\chi_f + \alpha\chi_e$\\
                \uIf(\tcp*[f]{non-saturating push}){$\U(e) - \sum_{\ell \in N} z_\ell(e) < w_i(e, f)$}{
                    $(j_e, g_e) \coloneqq (i, f)$ \tcp*[f]{set pointer to current pair}\\
                    continue with next undersold item $e$ (by aborting the for-loop)
                    \tcp*[f]{$e$ is not undersold anymore}
                }\Else(\tcp*[f]{saturating push}){
                    $(j_e, g_e) \coloneqq \begin{cases}
                        (i, f+1) \hspace{-2ex} & \text{if $f+1$ exists},\\
                        (i+1, 1) \hspace{-1ex}& \text{otherwise}
                    \end{cases}$
                    \tcp*[f]{set pointer to next pair in lex. order}
                }
            }
        }
        \If(\tcp*[f]{relabel if no push was found for $e$}){\Not $push$}{
            $\level(e) \coloneqq \level(e) + 1$\\
            $(j_e, g_e) \coloneqq (1, 1)$ \tcp*[f]{reset pointer to first pair in lex. order}
        }
    }
    \Return{$\bz$ and $\level(e)$ for all $e \in E$}
    \caption{}
    \label{alg:polymatroid_push_relabel_subroutine}
\end{algorithm}

\begin{theorem}\label{thm:polymatroid_runtime}
    The \ref{alg:polymatroid_push_relabel_subroutine} returns a solution to \eqref{eq:poly-union}.
    The algorithm runs in time $\mathcal{O}(n \cdot \BO + nm^3 \cdot \EO)$.
\end{theorem}
To prove this theorem, we first show correctness of the algorithm.
More precisely, we have to show the correctness of the procedure to select a buyer-item pair for a push.
In particular, we need to prove that with the described procedure, we do not relabel although there is still a feasible buyer-item pair for a push when ignoring buyer-item pairs which are lexicographically below the pointer.

\subsubsection{Correctness.}
We start by proving the above mentioned monotonicity lemma.

\begin{lemma}\label{lem:lexicographic_selection}
    Let $e \in E$ with $\level(e) = \ell$ and let $i \in N$.
    Under the assumption that the \ref{alg:polymatroid_union} executes all pushes with respect to a minimal suitable item,
    \begin{equation}\label{eq:order-increasement}
        \min \{f \mid f \in \Tt_i(e), \level(f) = \ell - 1\}
    \end{equation}
    is monotonically increasing and strictly increases at every saturating push at $e$ with respect to $f$ and $i$ as long as we do not relabel $e$, i.e., as long as $\level(e)$ remains $\ell$.
\end{lemma}
\begin{proof}
    We prove the lemma by considering the different steps a push and relabel algorithm can execute, namely pushing and relabeling and show that after both these operations, the statement of the lemma is fulfilled.
    A relabel operation does not affect $\Tt_i(e)$.
    Moreover, by \ref{i:l2}, there is no item $f \in \Tt_i(e)$ with $\ell - 2 = \level(f) < \level(e) - 1 = \ell - 1$.
    Thus, no additional candidates can appear in $\{f \mid f \in \Tt_i(e), \level(f) = \ell - 1\}$.
    The label of the currently minimal item could increase from $\ell-1$ to $\ell$ though, in which case \eqref{eq:order-increasement} increases.
    In total, we obtain the desired monotonicity for a relabel operation.

    For push operations with respect to $i' \neq i$, neither $\Tt_i(e)$ nor the labels change and thus, again, the statement holds since the minimal item remains the same.

    Consider a push operation at $k$ with respect to $i$ and $g$.
    Let $z'_i = z_i - \alpha \chi_g + \alpha \chi_k$ be the set assigned to $i$ after the push and let $\Tt'_i(e)$ denote the minimal $i$-tight set containing $e$ after the push.

    First assume that $k = e$.
    In this case, $g$ is the minimal item in $\Tt_i(e)$ with $\level(g) = \level(e)-1$ by the assumption of the lemma that every push is executed with respect to a minimal suitable item.
    \begin{itemize}
        \item For a  saturating push, we have $\alpha = w_i(e, g)$.
            This implies $z_i' - \chi_g + \chi_e \notin \B_i$ and thus, $g \notin \Tt'_i(e)$ by \Cref{lem:tight-sets-weights}.
            Moreover, $\rk_i(\Tt_i(e)) = z_i(\Tt_i(e)) = z_i'(\Tt_i(e))$ since $e, g \in \Tt_i(e)$.
            Hence, $\Tt_i(e)$ is a tight set with respect to $z_i'$ and, thus, $\Tt'_i(e) \subseteq \Tt_i(e)$.
            We get $\Tt'_i(e)\subseteq \Tt_i(e)\setminus\{g\}$, which implies that the item fulfilling all requirements is strictly increasing.

        \item For a non-saturating push, we get $\alpha < w_i(e, g)$.
            Thus, $z'_i - (w_i(e, g) - \alpha)\chi_g + (w_i(e, g) - \alpha)\chi_e \in  \B_i$, which implies $g \in \Tt'_i(e)$ by \Cref{lem:tight-sets-weights}.
            Moreover, $\Tt_i(e)$ is tight with respect to $z_i'$, hence, $\Tt'_i(e) \subseteq \Tt_i(e)$.
            This implies that the minimal item fulfilling all requirements remains the same.
    \end{itemize}
    Next, we show monotonicity for $k \neq e$.
    To this end, we distinguish three cases:
    \begin{itemize}
        \item If $k, g \notin \Tt_i(e)$, then $\Tt'_i(e) = \Tt_i(e)$ and thus, \eqref{eq:order-increasement} remains the same.
        \item If $k \in \Tt_i(e)$, it holds that $\Tt_i(k) \subseteq \Tt_i(e)$ by \Cref{lem:tight-sets-weights}.
            Thus, we get that also $g$ is in $\Tt_i(e)$, as $g \in \Tt_i(k) \subseteq \Tt_i(e)$.
            This implies that $z'_i(\Tt_i(e)) = z_i(\Tt_i(e))$.
            Since $\Tt_i(e)$ is a tight set, we have $ z_i(\Tt_i(e)) = \rk(\Tt_i(e))$.
            In other words, $\Tt_i(e)$ remains tight after the push.
            This yields $\Tt'_i(e) \subseteq \Tt_i(e)$, as $\Tt'_i(e)$ is defined as the minimal tight set containing $e$ and $\Tt_i(e)$ is already a valid candidate.
            Monotonicity of \eqref{eq:order-increasement} follows immediately.
        \item If $k \notin \Tt_i(e)$ and $g \in \Tt_i(e)$, it holds that $\Tt'_i(e) \subseteq \Tt_i(e) \cup \Tt_i(k)$.
            This follows since the set $\Tt_i(e) \cup \Tt_i(k)$ is tight and remains tight after the push since it contains $k$ and $g$.
            Now, we distinguish three cases, based on the level of $k$.
            First assume $\level(k) > \ell$.
            Then $\lmin(\Tt_i(k)) \geq \ell$ by \ref{i:l2} and therefore $\{f \in E \mid \level(f) = \ell-1\} \cap \Tt_i'(e) \subseteq \Tt_i(e)$.
            In other words, every item that is a candidate in the minimum after the push was also in $\Tt_i(e)$ and thus a candidate before the push since all items in $\Tt_i(k)$ have a label which is too high.
            Next, consider the case $\level(k) = \ell$.
            Then we get
            \begin{equation*}
                \min \{f' \mid f' \in \Tt_i(k), \level(f') = \ell - 1\} = g \geq \min\{f' \mid f' \in \Tt_i(e), \level(f') = \ell - 1\},
            \end{equation*}
            since we are now assuming $g \in \Tt_i(e)$ and $\level(g) = \ell-1$ and thus, $g$ is a valid candidate for the minimum on the right side.
            This shows the desired monotonicity of \eqref{eq:order-increasement} since we showed that $\Tt'_i(e) \subseteq \Tt_i(e) \cup \Tt_i(k)$.
            Finally, $\level(k) < \ell$ is not possible, as on one hand we have $\level(g) = \level(k)-1 < \ell - 1$ since we push on $k$ with respect to $g$ and on the other hand we have $g \in \Tt_i(e)$ by case distinction and thus, $\level(g) \geq \level(e) - 1 = \ell -1$.
            This is a contradiction.
    \end{itemize}
    Hence monotonicity of \eqref{eq:order-increasement} holds in every case.
\end{proof}

Note that when pushes are only executed with respect to lexicographically minimal item-buyer pairs, the condition of \Cref{lem:lexicographic_selection} is fulfilled.
Thus, we can apply the lemma to obtain monotonicity of \eqref{eq:order-increasement} for every buyer $i$.
This implies that, when fixing $i$, the potential items for a push at $e$ are monotonically increasing, i.e., not below the items considered so far.
Note that this implies that whenever there was no push for a buyer $i$ at item $e$, there will not be any push for this buyer until $e$ is relabeled.
Moreover, at a saturating push \eqref{eq:order-increasement}, increases strictly and thus also the buyer-item pair.
Hence, the procedure in~\ref{alg:polymatroid_push_relabel_subroutine} is indeed a valid implementation of the general push-relabel framework.

\begin{corollary}\label{cor:correctness_polymatroid_sum}
    The \ref{alg:polymatroid_push_relabel_subroutine} returns an optimal solution to~\eqref{eq:poly-union}.
\end{corollary}
Next, we aim to bound the running time.

\subsubsection{Running time analysis.}
As a first step to show \Cref{thm:polymatroid_runtime}, we bound the number of non-saturating pushes.
This is crucial since for these pushes, we cannot proceed with the next lexicographically higher item-buyer pair, but instead have to reconsider a pair that we looked at already.
The key ingredient of the proof is that we always choose an undersold item of maximal level below $m$ for a push and the fact that an item is not undersold anymore after a non-saturating push and can only get undersold again by a push on an item from a higher level.

\begin{lemma}\label{lem:number_pushes}
    The number of non-saturating pushes in the \ref{alg:polymatroid_push_relabel_subroutine} is at most $m^3$.
\end{lemma}
\begin{proof}
    A \emph{phase} of the algorithm is the set of iterations between two relabel operations.
    Once every item $e \in E$ has reached $\level(e) = m$, the algorithm stops since then there is no item left we can choose for a push.
    Since an item with level $m$ is never selected, there are at most $m^2$ phases ($m$ items, that are relabeled at most $m$ times).

    We will show that for each $e \in E$, there can be at most one non-saturating push in each phase.
    Here, we exploit that the algorithm picks an undersold item $e$ with $\level(e) = \ell$ maximal.
    After a non-saturating push, item $e$ is not undersold anymore.
    Thus, it cannot be picked again before it becomes undersold again.
    Push operations of an item at level $\ell'$ may create new undersold items but only on level $\ell'-1$.
    Therefore, item $e$ could become undersold again due to a push at an undersold item $f$ with $\level(f) = \ell + 1$.
    However, no undersold item at level $\ell+1$ may appear without a relabeling after $e$ was picked, i.e., within the same phase.

    Hence, there can be at most one non-saturating push per item, i.e., at most $m$ non-saturating pushes in a phase.
    As there are at most $m^2$ phases, there are at most $m^3$ non-saturating pushes during the algorithm.
\end{proof}

\begin{lemma}\label{lem:runtime_polymatroid_sum}
    The \ref{alg:polymatroid_push_relabel_subroutine} runs in $\mathcal{O}(n \cdot \BO + nm^3 \cdot \EO)$ time.
\end{lemma}
\begin{proof}
    We can find an initial allocation in $n \cdot \BO$ time by asking every buyer $i$ for a vector $z_i \in \B_i$.

    According to the previous lemma, the total number of non-saturating pushes is $\mathcal{O}(m^3)$.
    We estimate the running time of all relabel and saturating push operations induced by a fixed item $e$ while it is on level $\level(e) = \ell$.
    Given $e$, we go through all buyer-item combinations in a structured way such that we consider every buyer-item combination at most once unless there is a non-saturating push.
    So, without the non-saturating pushes an item $e$ on a fixed level $\ell$ induces a running time of $mn$ arithmetic operations and edge-weight queries, in particular one $\EO$ query for each item-buyer pair.
    Since there are $m$ items, that can be on at most $m$ different levels with a total addition of $m^3$ overall, the running time is $\mathcal{O}(nm^3 + m^3) = \mathcal{O}(nm^3)$ edge-weight queries and arithmetic operations.

    In the analysis so far, we ignored the time needed to pick an undersold item $e$ on the maximal possible level below $m$.
    We show now that this is not a bottleneck operation.
    We do so by using a standard network flow push-relabel construction (see e.g. \citep[Section 7.8]{ahuja1988network}).
    In the beginning we initialize a list per level giving all undersold items that are on this level.
    In the beginning, all lists are empty except for the level 0 list, which contains all undersold items.
    This construction takes $\mathcal{O}(m)$ but is only needed once.
    Moreover, we keep track of the maximal level that has a non-empty list and is below $m$.
    The list is easily maintainable, since only one item changes the level and the property undersold per iteration.
    The total increase of the maximal relevant level is bounded by $\mathcal{O}(m^2)$ ($m$~items on up to $m$ levels).
    Thus, the total decrease is at most $\mathcal{O}(m^2)$.
    Hence, scanning the lists to find the first non-empty list is not a bottleneck operation.
    Having such a data structure, we can access any item from the list on the level of the pointer in $\mathcal{O}(1)$.
\end{proof}

Now, we are ready to prove our main \Cref{thm:polymatroid_runtime} as we considered correctness and running time of the algorithm.

\begin{proof}[{Proof of \Cref{thm:polymatroid_runtime}}]
    The statement follows directly by \Cref{cor:correctness_polymatroid_sum} and \Cref{lem:runtime_polymatroid_sum}.
\end{proof}

Using the \ref{alg:polymatroid_push_relabel_subroutine} and a breadth-first search in the exchange graph, we can compute the minimal maximal over-/underdemanded sets efficiently.
\begin{corollary}
    Given prices $p$, we can compute the minimal maximal overdemanded [underdemanded] set in time $\mathcal{O}(n \cdot \BO + nm^3 \cdot \EO)$ in the multi-supply setting.
\end{corollary}
This follows by \Cref{lem:running_time_overd_set} [\Cref{lem:running_time_underd_set}, respectively] and \Cref{thm:polymatroid_runtime}.
Recall that the demand and exchange oracle calls are based on the matroids corresponding to the preferred bundles $\DM_i(p)$ or $\dhat_i(p)$, respectively.

\section{Minimal packing prices and maximal covering prices are Walrasian}
\label{sec:minimal-packing}

It is well-known that Walrasian prices form a complete lattice using the component-wise minimum and component-wise maximum as meet and join operations, respectively.
Hence, there also exists a component-wise minimal and maximal Walrasian price vector $\pmin$ and $\pmax$, which we already discussed in the previous sections.

Packing and covering prices are already interesting on their own as they both constitute weaker notions of equilibria (packing: every buyer gets a preferred bundle, covering: all items are sold).
Moreover, packing and covering prices are guaranteed to exist, also for non strong gross substitutes valuation functions, e.g., by setting the prices very high (such that no buyer is interested in any good) or to zero, respectively.
It is a natural question whether these relaxations inherit those lattice properties and whether minimal packing prices or maximal covering prices (if they exist) are guaranteed to be Walrasian.
Particularly, the question whether minimal packing and maximal covering prices exist is interesting.

For non strong gross substitutes valuations, we show in \Cref{appendix:packing_no_lattice} and \ref {appendix:covering_no_lattice} that packing and covering prices do not form a lattice.
Moreover, the example in \Cref{appendix:packing_no_lattice} even shows that there is no unique minimal packing price vector.

However, for strong gross substitutes valuations, the answer is positive: in the following, we show that there is a unique minimal packing price vector and a unique maximal covering price vector.
Moreover, we show that these vectors are equal to the minimal, respectively maximal, Walrasian price vector.

\subsection{Reduction to unit-supply case}
\label{sec:reduction}
Recall that we are given some market consisting of items (types) $e \in E$, where $b(e)$ is the number of units which are available of item $e$ and where each buyer $i \in N$ has a strong gross substitutes valuation function $v_i\colon \rng \to \Z_+$.
Moreover, we set prices $p(e)$ for each unit of item $e$.

Now, we will transfer this setting to a market with unit-supply, i.e., where only one unit of each item is available.
This unit-supply setting is used in the following proofs since with an individual price for every copy, it is easier to show the statements and we can directly use well-known results for the unit-supply case.
This simple trick of copying the items is also used for example in \citep[Appendix A.2]{murota2013computing}.

Given an instance in the multi-supply setting, we construct an instance in the unit-supply setting by defining $\U(e)$ copies $\tilde{e}_1, \dots, \tilde{e}_{\U(e)}$ of item $e \in E$.
We call the set of all copies of items $\tilde{E}$, i.e., $\tilde{E} \coloneqq \{\tilde{e}_1, \dots, \tilde{e}_{\U(e)} : e \in E\}$.
This means that $\tilde{e}_k$ is a copy of item $e$.
For simplicity, we may omit the index if we pick an item of $\tilde{E}$.
To convert the bundles $z$ to the unit-supply setting, we go through the buyers one by one and allocate any $z_i(e)$ (preferably unsold) different copies of $e$ to buyer $i$.
Note that if $z$ is packing there are enough unsold copies available, otherwise, we have to allocate some item copies to multiple buyers.
We convert a bundle $S \subseteq \tilde{E}$ to a bundle $z \in \rng$ by setting $z(e) = \left\vert \{\tilde{e} \in S \mid \tilde{e} \text{ is a copy of } e\}\right\vert$.
Thus, the valuation of $S$ is given by $v(S) \coloneqq v(z)$.

Note that in the unit-supply setting, it is allowed that two different copies $\tilde{e}_j$ and $\tilde{e}_k$ of item $e$ have different prices.
However, in the following lemma we show that in the minimal and maximal Walrasian price vector no two copies of an item have a different price.
Given this, converting a price vector from the unit-supply setting to the multi-supply setting becomes straight-forward.
For the reverse direction, given a multi-supply price vector $p$, we can construct a unit-supply price vector by just setting the price of each copy of $e$ to $p(e)$.
We stick to calling these prices $p$ since it will be clear from the context whether we are talking about the multi-supply or the unit-supply setting.

\begin{lemma}\label{lem:Walrasian-reduction}
	A vector $p$ is the minimal [maximal] Walrasian price vector in the multi-supply setting if and only if $p$ is also the minimal [maximal] Walrasian price vector in the corresponding unit-supply setting, i.e., the prices of each copy of an item $e$ are given by $p(e)$.
\end{lemma}
\begin{proof}
     It suffices to show that the minimal [maximal] Walrasian price vector $\tilde{p}$ in the unit-supply setting has the same prices for each copy of an item.
	 If two item copies $\tilde{e}_j$ and $\tilde{e}_k$ of item type $e$ do not have the same price, we will show that this price vector is not the minimal, resprectively maximal, Walrasian price vector.

     If $\tilde{p}(\tilde{e}_j) \neq \tilde{p}(\tilde{e}_k)$, we have that $\tilde{p}'$ with $\tilde{p}'(\tilde{f}) = \tilde{p}(\tilde{f})$ for $\tilde{f} \in \tilde{E} \setminus \{ \tilde{e}_j, \tilde{e}_k \}$, $\tilde{p}'(\tilde{e}_j) = \tilde{p}(\tilde{e}_k)$, and $\tilde{p}'(\tilde{e}_k) = \tilde{p}(\tilde{e}_j)$ is Walrasian by symmetry.
	 Hence, as Walrasian prices form a lattice \cite{gul1999walrasian}, $\tilde{p} \wedge \tilde{p}'$ and $\tilde{p} \vee \tilde{p}'$ are Walrasian as well.
	 This is a contradiction to $\tilde{p}$ being the minimal, respectively maximal, Walrasian price vector and thus, all item copies have the same price given the minimal [respectively maximal] Walrasian price vector $\tilde{p}$.
\end{proof}

The following lemma states that when using the described reduction to the unit supply setting (see \Cref{sec:reduction}), the valuation functions of the buyers remain gross substitutes.
\begin{lemma}[{\citep[Proposition A.1]{murota2013computing}}]\label{lem:GS-reduction}
    The valuation function $v\colon 2^{\tilde{E}} \to \Z_+$ is gross substitutes.
\end{lemma}

\subsection{Minimal packing prices are Walrasian}
It is already known that in the unit-supply setting, packing prices are also Walrasian if they lie below a Walrasian price vector \cite{ben2017walrasian}.
However, this is not sufficient to show that there exists a unique minimal packing price vector (which is Walrasian), as it is not a priori clear whether there can be a packing price vector that is in some component smaller (but in another larger) than the minimal Walrasian price vector.
In this section, we show that this is never the case as the minimal Walrasian price vector is a componentwise lower bound on every packing price vector.

\begin{theorem}\label{thm:min_packing_is_min_walrasian}
    Given an instance where all valuation functions are strong gross substitutes, let $\pmin$ be the minimal Walrasian price vector.
    Then $\pmin \leq q$ for any packing price vector $q$.
\end{theorem}

Remember that every Walrasian price vector is also packing, but not vice versa.
Thus, we can reformulate the main theorem of this section as follows.

\begin{corollary}
    There exists a component-wise minimal packing price vector and it is equal to the minimal Walrasian price vector.
\end{corollary}

To show \Cref{thm:min_packing_is_min_walrasian}, we will use the following lemma, that compares $\pmin(e)$ and $q(e)$ for the items $e$ which are at least partially sold in a packing allocation w.r.t.\ prices $q$.

\begin{lemma}\label{lemma:min_walrasian_leq_packing}
    Let $\pmin$ be the minimal Walrasian price vector and let $q$ be any packing price vector.
    Moreover, let $(z_1, \dots, z_n)$ be a packing allocation w.r.t.\ prices $q$.
    Then $\pmin(e) \leq q(e)$ for all $e \in E$ with $\sum_{i \in N} z_i(e) > 0$.
\end{lemma}
\begin{proof}
    To prove this lemma, we will use the reduction to the unit-supply case.
    Therefore, let $\tilde{E}$ denote the set containing all item copies and let $(S_1, \dots, S_n)$ be the sets of bundles in the unit-supply setting corresponding to $(z_1, \dots, z_n)$.

    By \Cref{lem:Walrasian-reduction} and since $\pmin$ is a minimal Walrasian price vector in the multi-supply setting, it is also a minimal Walrasian price vector in the unit-supply setting.
    Hence, in the unit-supply setting, it minimizes the submodular Lyapunov function $L(p) \coloneqq \sum_{i \in N} V_i(p) + \sum_{\tilde{e} \in \tilde{E}} p(\tilde{e})$ (see \Cref{lem:lyapunov}), i.e., $\pmin$ is an optimal solution to
    \begin{equation}\label{lyapunov_p}
        \alpha \coloneqq \min \left\{\sum_{i \in N} V_i(p) + \sum_{\tilde{e} \in \tilde{E}} p(\tilde{e}) \mid p(\tilde{e}) \geq 0 \text{ for all } \tilde{e} \in \tilde{E} \right\}.
    \end{equation}
    Let $S = \bigcup_{i \in N} S_i$ be the set of item copies chosen in the preferred bundles.
    Now, we define a new price vector $q'$ by
    \begin{equation*}
        q'(\tilde{e}) \coloneqq \begin{cases}
            q(\tilde{e}) & \text{if } \tilde{e} \in S,\\
            \infty & \text{otherwise}.
        \end{cases}
    \end{equation*}
    Note that $V_i(q) = V_i(q')$ holds by definition.
    The price vector $q$ restricted to $S$ is Walrasian.
    The same holds for $q'$, thus, it minimizes the Lyapunov function restricted to $S$.
    In other words, $q'$ is an optimal solution to
    \begin{equation}\label{lyapunov_q}
        \beta \coloneqq \min \left\{\sum_{i \in N} V_i(p) + p(S) \mid p(\tilde{e}) \geq 0 \text{ for all } \tilde{e} \in S \text{ and } p(\tilde{e}) = \infty \text{ for all } \tilde{e} \in \tilde{E} \setminus S \right\}.
    \end{equation}
    Next, we define the join $\check{p} \coloneqq \pmin \vee q'$ and the meet $\hat{p} \coloneqq \pmin \wedge q'$ of the two price vectors $\pmin$ and $q'$ (i.e., $\check{p}(\tilde{e}) \coloneqq \max\{\pmin(\tilde{e}), q'(\tilde{e})\}$ and $\hat{p}(\tilde{e}) \coloneqq \min\{\pmin(\tilde{e}), q'(\tilde{e})\}$).
    Then, the following holds:
    \begin{align*}
        \alpha + \beta &= \sum_{j \in N} V_j(\pmin) + \pmin(\tilde{E}) + \sum_{j \in N} V_j(q') + q'(S)
        = \sum_{j \in N} (V_j(\pmin) + V_j(q'))+(\pmin(\tilde{E}) + q'(S))\\
        &\geq \sum_{j \in N} (V_j(\hat{p}) + V_j(\check{p})) + (\hat{p}(\tilde{E}) + \check{p}(S)) = \sum_{j \in N} V_j(\hat{p}) + \hat{p}(\tilde{E}) + \sum_{j \in N} V_j(\check{p}) + \check{p}(S)\\
        &\geq \alpha + \beta,
    \end{align*}
    where the first inequality follows from submodularity of $p \mapsto V_j(p)$ (cf. \citep[Theorem 10]{ausubel2002ascending}) and since $\pmin(\tilde{e}) + q'(\tilde{e}) = \check{p}(\tilde{e}) + \hat{p}(\tilde{e})$ for all $\tilde{e} \in S$ and $\hat{p}(\tilde{e}) = \pmin(\tilde{e})$ for all $\tilde{e} \in \tilde{E} \setminus S$.
    The second inequality follows since $\hat{p}$ is a feasible price vector for \eqref{lyapunov_p} and since $\check{p}$ is a feasible price vector for \eqref{lyapunov_q}.
    Therefore, equality must hold in the equation system.
    This implies that $\hat{p}$ is an optimal solution for \eqref{lyapunov_p} and thus a Walrasian price vector.
    Since $\pmin$ is the minimal Walrasian price vector, it holds that $\pmin(\tilde{e}) \leq \hat{p}(\tilde{e}) = \min \{\pmin(\tilde{e}), q'(\tilde{e})\} \leq q'(\tilde{e}) = q(\tilde{e})$ for $\tilde{e} \in S$.
    By definition of $S$, it follows that $\pmin(e) \leq q(e)$ holds for all $e \in E$ with $\sum_{i \in N} z_i(e) > 0$.
\end{proof}

\begin{proof}[Proof of \Cref{thm:min_packing_is_min_walrasian}]
Let $q$ be a minimal packing price vector and consider a packing allocation $(z_1, \ldots, z_n)$ w.r.t.\ prices $q$.
    By \Cref{lemma:min_walrasian_leq_packing}, we already know that $\pmin(e) \leq q(e)$ for all $e$ with $\sum_{i \in N} z_i(e) > 0$.
    For $e$ with $z_i(e) = 0$ for all $i \in N$, we consider two cases.

    If $z_i'(e) > 0$ for some $z_i' \in \dset_i(q)$ for some buyer $i \in N$, we show that there is $z_i'' \in \dset_i(q)$ with $z''(e)>0$ and $z_i'' \leq z_i + \chi_e$:
    Since both $z_i$ and $z_i'$ are preferred bundles of buyer $i$, they maximize the utility function $u_i$.
    Since the utility function $u_i$ is \Mnat-concave, for $e \in \supp^+(z_i' - z_i)$ one of the following cases holds (see \Cref{def:Mnat_concave}):
    \begin{enumerate}[label=(\alph*)]
        \item It holds that $u_i(z_i') + u_i(z_i) \leq u_i(z_i' - \chi_e) + u_i(z_i + \chi_e)$.
            Hence, equality holds since $z_i'$ and $z_i$ are maximizers of the utility and thus, $z_i'' = z_i+\chi_e$ is also a preferred bundle.
        \item There exists an $f \in \supp^-(z_i' - z_i)$ with $u_i(z_i') + u_i(z_i) \leq u_i(z_i' - \chi_e + \chi_f) + u_i(z_i + \chi_e - \chi_f)$.
            Again, equality must hold and thus $z_i'' = z_i + \chi_e - \chi_f$ is a preferred bundle.
    \end{enumerate}
    In both cases, there exists a packing allocation $\bar{\bz} = (z_1, \dots, z_{i-1}, z_i'', \dots, z_n)$ with respect to prices $q$ with $\sum_{i \in N}\bar z_i(e) > 0$.
    Thus, it follows by \Cref{lemma:min_walrasian_leq_packing} that $\pmin(e) \leq q(e)$.

    If $e \notin \supp^+(z_i')$ for any $z_i' \in \dset_i(q)$ and any buyer $i \in N$, we can reduce the price of $e$ slightly\footnote{Note, that we are talking about preferred bundles which do not need to be minimal.
	Hence, if $p(e)=0$, it can be added to every preferred bundle of every buyer.}, which results in a smaller packing price vector and thus in a contradiction to the choice of~$q$.
\end{proof}

\subsection{Maximal covering prices are Walrasian}
In this section, we establish that maximal covering prices are Walrasian.
More precisely, we show that first of all, there exists a unique component-wise maximal covering price vector, and second that this price vector coincides with the maximal Walrasian price vector.

The main result of this section can be formulated as follows:
\begin{theorem}\label{thm:max_market_clearing_is_max_walrasian}
	Given an instance where all valuation functions are strong gross substitutes, let $\pmax$ be the maximal Walrasian price vector.
	Then $\pmax \geq q$ for any covering price vector $q$.
\end{theorem}

\begin{corollary}
	There exists a component-wise maximal covering price vector and it is equal to the maximal Walrasian price vector.
\end{corollary}

The remainder of this section is dedicated to the proof of \Cref{thm:max_market_clearing_is_max_walrasian}.
By \Cref{lem:Walrasian-reduction}, it suffices to establish \Cref{thm:max_market_clearing_is_max_walrasian} in the unit-supply setting because for any covering vector $q$ in the multi-supply setting, $q$ will also be covering in the corresponding unit-supply instance.

Fix a unit-supply instance with item set $E$ and buyers $i \in N$ with gross substitutes valuation functions $v_i\colon 2^E \to \Z_{+}$.
Denote the maximal Walrasian price vector by $\pmax$ and let $q$ be covering.
Fix preferred bundles $(S_i)_{i \in N}$ for prices $q$ with $\bigcup_{i \in N} S_i = E$.
For $e \in E$, let $k_e \coloneqq |\{i \in N \mid e \in S_i\}|$ denote the multiplicity of item $e$ among the bundles $(S_i)_{i \in N}$.

Our strategy to prove \Cref{thm:max_market_clearing_is_max_walrasian} is to construct an instance in which we have $k_e$ copies of each item, and argue that $q$ will be Walrasian prices for that instance (assigning to each copy of $e$ a price of $q(e)$).

To this end, let $E' \coloneqq \{e'_1, \ldots, e'_{k_e} \mid e \in E\}$ and let $\pi\colon E' \to E, e'_i \mapsto e$ be the projection of $E'$ onto $E$.
We define new valuation functions $v'_i\colon E' \to \Z_{+}$ for $i \in N$ by setting
\[
     v'_i(S') \coloneqq v_i(\pi(S')).
\]
For a price vector $p \in \R_+^E$ [$p' \in \R_+^{E'}$], we denote the set of preferred bundles of buyer $i$ w.r.t. prices $p$ [prices $p'$] by $\dset_i(p)$ [$\dset'_i(p')$].

For a price vector $p' \in \R_+^{E'}$, we denote by $p^\downarrow \in \R_+^E$ the price vector with
\[
    p^\downarrow(e) \coloneqq \min \{p'(e'_i) \mid i \in \{1,\dots, k_e \} \}.
\]
\begin{lemma}\label{lemma:connection_demand_sets}
	Let $p'\in \R_+^{E'}$ and $S' \subseteq E'$.
	  Then $S' \in \dset_i'(p')$ for an $i \in N$ if and only if
	\begin{itemize}
		\item $\pi(S') \in \dset_i(p^\downarrow)$,
		\item $S'$ only contains cheapest copies of each item, and
		\item whenever $S'$ contains multiple copies of an item, all of them have price $0$.
	\end{itemize}
\end{lemma}
\begin{proof}
	Let $p' \in \R^{E'}_+$ and let $S' \subseteq E'$ be an arbitrary set and $S \coloneqq \pi(S')$ its projection.
    Further, let $T \in \dset_i(p^\downarrow)$ and let $T'$ arise from $T$ by picking a cheapest copy of each item in $T$.
	Then
	\begin{align*}
		v'_i(S') - p'(S') &= v_i(S) - p'(S') \stackrel{(*)}{\leq} v_i(S)-p^\downarrow(S) \stackrel{(**)}{\leq} v_i(T)-p^\downarrow(T) = v'_i(T')-p'(T').
	\end{align*}
    The first inequality holds by definition of $p^\downarrow$ and the second inequality holds by definition of $T$.
	This chain of inequalities shows that $T' \in \dset_i'(p')$ because $S' \subseteq E'$ was chosen arbitrarily.
    Knowing this, we have $S' \in \dset'_i(p')$ if and only if the outer inequality is tight, which is the case if and only if $(*)$ and $(**)$ are tight.

	$(*)$ is tight if and only if $p'(S') = p^\downarrow(S)$, i.e., $S'$ only contains cheapest copies of any item and whenever $S$ contains multiple copies of an item, all of them have price $0$.
    $(**)$ is tight if and only if $S \in \dset_i(p^\downarrow)$.
\end{proof}
\begin{lemma}
    The valuation function $v'_i$ is gross substitutes for each $i \in N$.
\end{lemma}
\begin{proof}
    By \Cref{lem:GS-Mconv}, it suffices to establish that $v'_i$ is \Mnat-concave for each $i \in N$.

	Let $X', Y' \subseteq E'$ and let $X \coloneqq \pi(X')$ and $Y \coloneqq \pi(Y')$.
	Let further $e' \in X' \setminus Y'$ and define $e \coloneqq \pi(e')$.
    Then $e \in X$.

    \begin{description}[leftmargin=2ex]
        \item[Case 1:] It holds that $e \notin Y$.\\
            Then, as $v_i$ is gross substitutes, it is \Mnat-concave (see \Cref{lem:GS-Mconv}), so at least one of the equations \ref{i:M1} or \ref{i:M2} holds:

            \begin{description}[leftmargin=2ex]
                \item[Case 1.1:] Equation \ref{i:M1} holds, i.e., $v_i(X) + v_i(Y) \leq v_i(X \setminus \{e\}) + v_i(Y \cup \{e\})$.\\
                    As $X \setminus \{e\} \subseteq \pi(X' \setminus \{e'\})$ and $\pi(Y' \cup \{e'\}) = Y \cup \{e\}$, it holds by monotonicity of $v_i$ that
                    $v'_i(X') + v'_i(Y') = v_i(X) + v_i(Y) \leq v_i(X \setminus \{e\}) + v_i(Y \cup \{e\}) \leq v'_i(X' \setminus \{e'\}) + v'_i(Y' \cup \{e'\}).$
		          \item[Case 1.2:] Equation \ref{i:M2} holds, i.e., $v_i(X) + v_i(Y) \leq v_i(X \setminus \{e\} \cup \{f\}) + v_i(Y \setminus \{f\} \cup \{e\})$ for an $f \in Y \setminus X$.
                    Let $f' \in Y'$ with $\pi(f') = f$.
                    Then $f' \in Y' \setminus X'$.
                    Moreover, $X \setminus \{e\} \cup \{f\} \subseteq \pi(X' \setminus \{e'\} \cup \{f'\})$ and $Y \setminus \{f\} \cup \{e\} \subseteq \pi(Y' \setminus \{f'\} \cup \{e'\})$.
                    Hence, by monotonicity of $v_i$, we get
		              \begin{align*}
                        v'_i(X') + v'_i(Y') &= v_i(X) + v_i(Y) \leq v_i(X \setminus \{e\} \cup \{f\}) + v_i(Y \setminus \{f\} \cup \{e\})\\
	       		          &\leq v'_i(X' \setminus \{e'\} \cup \{f'\}) + v'_i(Y' \setminus \{f'\} \cup \{e'\}).
                    \end{align*}
            \end{description}
            Hence, either \ref{i:M1} or \ref{i:M2} holds for $v_i'$.
        \item[Case 2:] It holds that $e \in Y$.
            \begin{description}[leftmargin=2ex]
                \item[Case 2.1:] There exists an $e'' \in X' \setminus \{e'\}$ with $\pi(e'') = e$.\\
                Then we have $\pi(X' \setminus \{e'\}) = \pi(X')=X$ and $\pi(Y' \cup \{e'\}) = \pi(Y')=Y$.
                This yields
                \[
                    v'_i(X') + v'_i(Y') = v_i(X) + v_i(Y)
                    = v'_i(X' \setminus \{e'\}) + v'_i(Y' \cup \{e'\}).
                \]
    	        \item[Case 2.2:] There is no $e'' \in X' \setminus \{e'\}$ with $\pi(e'') = e$.\\
                    As $e \in Y$, there exists an $f' \in Y'$ with $\pi(f') = e$.
                    It holds that $e' \notin Y'$, so $f' \neq e'$.
					In particular, it follows that $f' \notin X'$.
                    Hence, $f' \in Y' \setminus X'$.
                    We further have $\pi(X' \setminus \{e'\} \cup \{f'\}) = X$ and $\pi(Y' \setminus \{f'\} \cup \{e'\}) = Y$.
                    This results in
    	            \[
                        v'_i(X') + v'_i(Y') = v_i(X) + v_i(Y)
	       	            = v'_i(X' \setminus \{e'\} \cup \{f'\}) + v'_i(Y' \setminus \{f'\} \cup \{e'\}).
                    \]
            \end{description}
            Hence, also either \ref{i:M1} or \ref{i:M2} holds for $v_i'$.
	\end{description}
    Thus, $v_i'$ is \Mnat-convex and thus gross substitutes by \Cref{lem:GS-Mconv}.
\end{proof}

\begin{lemma}
    \label{lem:qprime_is_Walrasian}
    Let $q \in \R^E_+$ be a covering price vector.
	Then $q' \in \R^{E'}_+$ with $q'(e') \coloneqq q(\pi(e'))$ for $e' \in E'$ is a Walrasian price vector with respect to supply $E'$ and valuation functions $v'_i$ for $i \in N$.
\end{lemma}
\begin{proof}
    By our definition of the numbers $(k_e)_{e \in E}$, we can partition $E'$ into sets $(S'_i)_{i \in N}$ with $\pi(S'_i) = S_i$ and $|S'_i| = |S_i|$ for $i \in N$.
    By \Cref{lemma:connection_demand_sets}, we have $S'_i \in \dset'_i(q')$ for all $i \in N$.
\end{proof}

Now, we are ready to prove the main theorem of this section.

\begin{proof}[Proof of \Cref{thm:max_market_clearing_is_max_walrasian}]
    As mentioned before, it suffices to show this statement in the unit-supply setting, i.e., we assume $\U(e)=1$ for all $e \in E$.
    Let $\pmax$ be the maximal Walrasian price vector and let $q$ be any covering price vector.

    In the following, we will use that Walrasian prices are minimizers of the Lyapunov function (see \Cref{lem:lyapunov}).
    We distinguish the settings by using $L$ for the Lyapunov function with respect to $E$ and $v_i$ for all $i \in N$, and $L'$ for the Lyapunov function with respect to $E'$ and $v'_i$.
    Analogously, let $V_i(p) = \max_{S \subseteq E} v_i(S) - p(S)$ and $V'_i(p') = \max_{S' \subseteq E'} v'_i(S') - p'(S')$.

    We know that $\pmax$ is a minimizer of the Lyapunov function in the original setting, i.e., it minimizes
    \[
        L(p) = \sum_{i \in N} V_i(p) + p(E).
    \]
    Moreover, $q'$ with $q'(e') \coloneqq q(\pi(e'))$ for $e' \in E'$ is a Walrasian price vector for $E'$ and $(v'_i)_{i \in N}$ (see \Cref{lem:qprime_is_Walrasian}), i.e., it minimizes
    \[
        L'(p') = \sum_{i \in N} V'_i(p') + p'(E').
    \]
    Note that $q^\downarrow = q$ by definition of $q'$.
	By \Cref{lemma:connection_demand_sets}, we have
    \[
        V_i'(q') = \max_{S' \subseteq E'} v'_i(S') - q'(S') = \max_{S \subseteq E} v_i(S) - q^\downarrow(S) = V_i(q^\downarrow) = V_i(q)
    \]
    for all $i \in N$.
    Thus,
    \[
        L'(q') = \sum_{i \in N} V_i(q) + q'(E').
    \]
    Define $\check{p} \in \R_+^E$ and $\hat{q}' \in \R_+^{E'}$ via
    \[
        \check{p}(e) \coloneqq \max \{\pmax(e), q(e)\} \qquad \text{ and } \qquad
        \hat{q}'(e'_\ell) \coloneqq \begin{cases}
        	\min \{\pmax(e), q(e)\} & \text{for }\ell = 1,\\
            q(e) & \text{otherwise}.
        \end{cases}
    \]
    This yields
    \begin{align*}
        L(\pmax) + L'(q') &= \sum_{i \in N} V_i(\pmax) + \pmax(E) + \sum_{i \in N} V_i(q) + q'(E')\\
        &= \sum_{i \in N} \left( V_i(\pmax) + V_i(q) \right) + \sum_{e \in E} \left( \pmax(e) + k_e \cdot q(e) \right)\\
        &\geq \sum_{i \in N} \left( V_i(\max \{\pmax, q\}) + V_i(\min \{\pmax, q\}) \right)\\
            &\phantom{\geq}\ + \sum_{e \in E} \left( \max \{\pmax(e), q(e)\} + \min \{\pmax(e), q(e)\} + (k_e-1) \cdot q(e) \right)\\
        &= \sum_{i \in N} V_i(\check{p}) + V_i(\hat{q}'^\downarrow) + \check{p}(E) + \hat{q}'(E')\\
        &= L(\check{p}) + L'(\hat{q}'),
    \end{align*}
    where in the first inequality, we used that $V_i$ is lattice submodular (cf. \citep[Theorem 10]{ausubel2002ascending}).
    Hence, $\check{p}$ and $\hat{q}'$ are minimizers of $L$ and $L'$, respectively, as well.
    Thus, $\check{p}$ is Walrasian.
    As $\pmax$ is the maximal Walrasian price vector, $q(e) \leq \max \{\pmax(e), q(e)\} = \check{p}(e) \leq \pmax(e)$ for all $e \in E$.
    This yields $q \leq \pmax$.
\end{proof}

\section{Monotonicity analysis}\label{sec:sensitivity}
In this section, we show monotonicity results in supply and demand for the (unique) buyer-optimal, as well as for the (again unique) seller-optimal Walrasian prices, provided that all buyers' valuations are strong gross substitutes.
Independently, \citet{raach2024monotonicity} proved the same monotonicity results.
We also provide an example showing that monotonicity cannot be guaranteed if the valuation functions are changed.
We model the decrease in supply by restricting the market to less copies of items $b' \leq b$ with $b'(e) < b(e)$ for at least one $e \in E$ and by restricting all valuations to $v_i' = v_i|_{[{\bf0},\U']_\Z}$.

The decrease in demand is modeled as item-truncation, i.e., we assume that a buyer $i$ is only interested in $d_i$ items in total.
Since it is known that strong gross substitutes functions are closed under item-truncation (cf.~\citep{collina2020approximability}\footnote{\citet{collina2020approximability} showed it in the unit-supply setting.
The statement for the multi-supply setting follows by copying the items as presented in \Cref{sec:reduction}.}), all valuation functions remain (strong) gross substitutes after this operation.
Formally, this leads to the following definition.
\begin{definition}
    Let $v_i \colon \rng \to \Z_+$ be a valuation function and $d \in \Z^N_+$.
    The \emph{truncation} $v_i^d$ of buyer $i$ to demand $d$ is given by $v_i^d(z) \coloneqq \max\{v_i(y) \mid y \leq z, \sum_{e\in E} y(e) \leq d_i\}$.
\end{definition}

We write $\dset_i^d(p)$ [respectively $\DM_i^d(p)$ and $\dhat_i^d(p)$] to denote the set of [minimal/maximal] preferred bundles of buyer $i$ at prices $p$ with valuation function $v_i^d$.
Note that if we choose $d_i$ large enough, it holds that $v_i(z) = v_i^{d}(z)$ for all $z \in \rng$.
Hence, without loss of generality, we assume that each buyer indeed has a demand and denote the demand vector by $d = (d_1, \dots, d_n)$.

The monotonicity properties we want to show are the following:
\begin{itemize}
    \item \emph{in supply:}
        Let $p$ be the buyer-optimal [seller-optimal] Walrasian price vector.
        If the supply of some items decreases, i.e., if there are only $b' \leq b$ items to be sold, then the corresponding buyer-optimal [seller-optimal] Walrasian prices $p'$ for items $e \in E' = \{ f \in E \mid \U'(f) >0 \}$ will only increase, i.e., $p(e) \leq p'(e)$ for all $e \in E'$.
    \item \emph{in demand:}
        Consider two demand vectors $d'$ and $d$ in $\Z_+^N$ with $d' \leq d$, and let $p$ and $p'$ be the buyer-optimal [seller-optimal] Walrasian price vectors if the valuation functions $v_i\colon \rng \to \Z_+$ of the buyers $i$ are truncated at demands $d$ and $d'$, respectively.
        Then the buyer-optimal [seller-optimal] Walrasian prices $p'$ will only decrease if the demand decreases, i.e., $p' \leq p$.
\end{itemize}

\begin{theorem}\label{thm:monotonicity}
    Monotonicity in supply and demand with respect to the minimal and the maximal Walrasian price vector can be guaranteed if all buyers' valuations are strong gross substitutes.
\end{theorem}

The main idea of the proof is to use the equivalence between minimal packing [maximal covering] and minimal [maximal] Walrasian prices.
We first show monotonicity in supply.

\begin{lemma}\label{lem:sensitivity_supply}
    Monotonicity in supply with respect to the minimal Walrasian price vector can be guaranteed if all buyers' valuations are strong gross substitutes.
\end{lemma}

\begin{proof}
    Let $b$ and $b'$ be two supply vectors with $b' \leq b$.
    Assume without loss of generality that $\U$ and $\U'$ only differ in one item $f$ and only in one unit, i.e., $b(f)-b'(f)=1$ and $b(e)=b'(e)$ for all $e \in E \setminus \{f\}$.
    We fix a Walrasian allocation $(z_1,\dots,z_n)$ given the minimal Walrasian prices $p'$ at supply $E'$.
    Now, we use the reduction to the unit-supply setting.
    Therefore, let $\tilde{E}$ denote the set containing all item copies with respect to $\U$.
    Further let $(S_1,\dots,S_n)$ be the packing allocation corresponding to $(z_1,\dots,z_n)$.

    We fix an arbitrary copy $\tilde{f} \in \tilde{E}$ of the unique item $f \in \supp^+(\U - \U')$.
    Next, we adapt the prices $p'$ to prices $\bar{p}$ by setting $\bar{p}(\tilde{e}) = p'(\tilde{e})$ for $\tilde{e} \in \tilde{E}\setminus \{\tilde{f}\}$ and $\bar{p}(\tilde{f}) = \max_{i \in N} v_i(E) + 1$.
    Thus, no buyer wants to buy item $\tilde{f}$, so any packing allocation with respect to supply $\tilde{E} \setminus \{\tilde{f}\}$ and prices $p'$ is also packing with respect to supply $\tilde{E}$ and prices $\bar{p}$.
	As there is a packing allocation in the first setting, we know that the same allocation is packing for prices $\bar{p}$ at supply $\tilde{E}$.
    Using that $p$ is the component-wise minimal packing price vector at supply $\tilde{E}$, we get by \Cref{thm:min_packing_is_min_walrasian} that $p(\tilde{e}) \leq \bar{p}(\tilde{e}) = p'(\tilde{e})$ for all $\tilde{e} \in \tilde{E} \setminus \{\tilde{f}\}$.

    Thus, $p(e)\leq p'(e)$ holds for every item $e \in E$ with a copy in $\tilde{E} \setminus \{\tilde{f}\}$, i.e., for the items $e \in E$ with $\U'(e)>0$.
\end{proof}

\begin{lemma}\label{lem:sensitivity_supply_max}
	Monotonicity in supply with respect to the maximal Walrasian price vector can be guaranteed if all buyers' valuations are strong gross substitutes.
\end{lemma}
\begin{proof}
    As before, let $b$ and $b'$ be two supply vectors with $b' \leq b$ and assume without loss of generality that there is $f \in E$ such that $b(f) = b'(f) + 1$ and $b(e) = b'(e)$ for all $e \in E \setminus \{f\}$.
    Let $(z_1, \dots, z_n)$ be a Walrasian allocation for the maximal Walrasian prices $p$ at supply $b$.
    Again, we reduce to the unit-supply setting.
    Let $\tilde{E}$ denote the set containing all item copies with respect to $b$ and let $\tilde{p} \in \R_+^{\tilde{E}}$ be the price vector with $\tilde{p}(\tilde{e}) = p(e)$ for every copy $\tilde{e}$ of $e \in E$.
    Let further $(S_1, \dots, S_n)$ be the partition of $\tilde{E}$ corresponding to $(z_1, \dots, z_n)$.

    We fix an arbitrary copy $\tilde{f}$ of $f$ and define a price vector $\bar{p} \in \R_+^{\tilde{E}}$ by setting $\bar{p}(\tilde{e}) = \tilde{p}(\tilde{e})$ for all $\tilde{e} \in \tilde{E} \setminus \{\tilde{f}\}$ and $\bar{p}(\tilde{f}) = \max_{i \in N} v_i(E) + 1$.
    As the valuation functions $(\tilde{v}_i)_{i \in N}$ are gross substitutes, we know that there exist preferred bundles $(\bar{S}_i)_{i \in N}$ with respect to $\bar{p}$ such that $S_i \setminus \{\tilde{f}\} \subseteq \bar{S}_i$.
    Moreover, by our choice of $\bar{p}(\tilde{f})$, we can infer that $\tilde{f} \notin \bar{S}_i$ for $i \in N$.
    In particular, the bundles $(\bar{S}_i)_{i \in N}$ also constitute preferred bundles for supply $\tilde{E} \setminus \{\tilde{f}\}$ and the restriction of $\tilde{p}$ to $\tilde{E} \setminus \{\tilde{f}\}$.
    Furthermore, $\bigcup_{i \in N} S_i = \tilde{E}$ implies $\bigcup_{i \in N} \bar{S}_i = \tilde{E} \setminus \{\tilde{f}\}$.

    Going back to the multi-supply setting, this tells us that the prices $p$ are covering for supply $b'$.
    By \Cref{thm:max_market_clearing_is_max_walrasian}, we may conclude $p\leq p'$ where $p'$ is the minimal Walrasian price vector for supply $b'$.
\end{proof}

Now, we examine the effect of changes in the buyer's demand on the minimal [maximal] Walrasian prices.
\begin{lemma}\label{monotone_in_demand}
    Monotonicity in demand with respect to the minimal Walrasian price vector can be guaranteed if all buyers' valuations are strong gross substitutes.
\end{lemma}

\begin{lemma}\label{monotone_in_demand_max}
	Monotonicity in demand with respect to the maximal Walrasian price vector can be guaranteed if all buyers' valuations are strong gross substitutes.
\end{lemma}
Clearly, \Cref{thm:monotonicity} follows directly form \Cref{lem:sensitivity_supply}, \Cref{lem:sensitivity_supply_max}, \Cref{monotone_in_demand} and \Cref{monotone_in_demand_max}.
To show the last two lemmas, we will use that in the unit-supply setting, gross substitutes valuations are well-layered (cf.~\citep{dress1995well,leme2017gross}).
Intuitively, well-layered means that the greedy algorithm (always selecting an item with the highest marginal value) computes in each iteration a bundle with the highest utility among all bundles of that size.

\begin{definition}[\citep{dress1995well}]
    In the unit-supply setting, let $v\colon 2^E \to \Z_+$ and $p\in \Z^E$.
	A \emph{greedy sequence for $v$ and $p$} is an ordering $e_1,\dots,e_{|E|}$ of the elements of $E$ such that for every $j\in \{1,\dots,|E|\}$, we have
\[e_j\in \arg\max_{e \notin \{e_1,\dots,e_{j-1}\}} \{v(\{e_1,\dots,e_{j-1},e\})-v(\{e_1,\dots,e_{j-1}\})-p(e)\}.\]
\end{definition}
Based on our reduction to the unit-supply case, we obtain the following analogous notion of a greedy sequence for the multi-supply setting:
\begin{definition}
	In the multi-supply setting, let $v \colon \rng \to \Z_+$ and $p\in\Z^E$.
	A \emph{greedy sequence} for $v$ and $p$ is a sequence $e_1,\dots,e_{b(E)}$ of elements of $E$ with the following properties:
	For $j\in\{0,\dots,b(E)\}$, let $z^j\in \rng$ such that $z^j(e)$ equals the number of occurrences of $e$ among $e_1,\dots,e_j$.
	Then $z^{b(E)}=b$ and
	for each $j\in \{1,\dots,b(E)\}$, we have $e_j \in \arg\max_{e \in \supp^+(b-z^{(j-1)})} \{v (z^{(j-1)} + \chi_e) - v(z^{(j-1)}) - p(e)\}$.
\end{definition}
\begin{definition}[\citep{dress1995well}]
A valuation function $v\colon 2^E \to \Z_+$	is called \emph{well-layered} if for any price vector $p\in \Z^E$ and any greedy sequence $e_1,\dots,e_{|E|}$ for $v$ and $p$, we have
\[\{e_1,\dots,e_j\}\in \arg\max\{v(S)-p(S)\colon S\subseteq E,|S|=j\} \text{ for all }j\in\{1,\dots,|E|\}.\]
\end{definition}

\begin{lemma}[\citep{dress1995well}]
    In the unit-supply setting, a non-decreasing gross substitutes function $v\colon 2^E \to \Z_+$ is \emph{well-layered}.
\end{lemma}
For non-decreasing gross substitutes valuations, it is further known that \emph{every} bundle that is optimum for a given size can be obtained via a greedy sequence.
\begin{lemma}[{\cite[Remark 6.5]{leme2017gross}}]
    In the unit-supply setting, let $v\colon 2^E \to \Z_+$ be a non-decreasing gross substitute	valuation function and let $p\in\Z^E$.
	Let further $j\in\{1,\dots,|E|\}$ and $S^*\in\arg\max\{v(S)-p(S)\colon S\subseteq E,|S|=j\}$.
	Then there exists a greedy sequence $e_1,\dots,e_{|E|}$ for $v$ and $p$ such that $S^*=\{e_1,\dots,e_j\}$.
\end{lemma}
In the multi-supply setting, we can obtain a similar result by using the reduction to the unit supply case as described in \Cref{sec:reduction}, in particular, that the corresponding valuation functions are gross substitutes and thus well-layered.
\begin{corollary}\label{cor:well-layered}
    In the multi-supply setting, for a strong gross substitutes valuation function $v$, demand $d$ and prices $p \in \Z_+^E$, and a preferred bundle $z \in \dset_i^d(p)$, there exist a greedy sequence $e_1,\dots,e_{b(E)}$ for $v^d$ and $p$, as well as $j^*\in\{0,\dots,b(E)\}$, such that
    \begin{itemize}
    	\item $z^{j^*}=z$,
    	\item $\max_{e \in \supp^+(b-z^{(j-1)})} \{v (z^{(j-1)} + \chi_e) - v(z^{(j-1)}) - p(e)\}\geq 0$ for all $j\in\{1,\dots,j^*\}$ and
    	\item $\max_{e \in \supp^+(b-z^{(j-1)})} \{v (z^{(j-1)} + \chi_e) - v(z^{(j-1)}) - p(e)\}\leq 0$ for all $j\in\{j^*+1,\dots,b(E)\}$.
    \end{itemize}
\end{corollary}
With this knowledge, we can prove \Cref{monotone_in_demand}.

\begin{proof}[Proof of \Cref{monotone_in_demand}]

    Let $d$ and $d'$ be two demand vectors with $d' \leq d$.
    Let $\pmin$ and $\pmin'$ be the minimal Walrasian prices for demand $d$ and $d'$, respectively, and let $(z_1, \dots, z_n)$ be a Walrasian allocation for prices $\pmin$ and demand $d$.

    We will show that for each buyer $i$, there is a vector $y_i \leq z_{i}$ with $y_i \in \dset_{i}^{d'}(\pmin)$.
    If $z_{i}(E) \leq d'_{i}$ we are done (we set $y_i=z_i$).
    Thus, we assume $z_i(E) > d'_i$ in the sequel.
    Since $v_i^d$ is strong gross substitutes and $z_i \in \dset_{i}^{d}(\pmin)$, we can obtain $z_{i}$ using the greedy algorithm (\Cref{cor:well-layered}).
    Let $y_i$ be the bundle of the $d'_{i}$ items that were selected first.
    Hence, $y_i \leq z_i$.
    Moreover, all marginal returns in the greedy sequence leading to $y_i$ are non-negative and by definition of $v_i^{d'}$, every further item has a non-positive marginal utility w.r.t.\ $y_i$.
	Hence, $y_i \in \dset_{i'}^{d'}(\pmin)$.

    Thus, for demand $d'$, we know that $(y_1, \dots, y_n)$ is a packing allocation at prices $\pmin$.
    Hence, $\pmin$ is a packing price vector.
    Since $\pmin'$ is the buyer-optimal Walrasian price vector at demand $d'$,
    it holds, by \Cref{thm:min_packing_is_min_walrasian}, that $\pmin' \leq \pmin$.

\end{proof}

\begin{proof}[Proof of \Cref{monotone_in_demand_max}]
Let $d$ and $d'$ be two demand vectors with $d'\leq d$.
Let $p'$ be Walrasian prices for demand $d'$ and let $(z'_1,\dots,z'_n)$ be a Walrasian allocation for prices $p'$ and demand $d'$.
By \Cref{cor:well-layered}, we can extend each bundle $z'_i$ to a preferred bundle $z_i\in\dset_i^d(p')$ with $z_i\geq z'_i$ (component-wise) by continuing the greedy allocation process until the marginal returns become negative.
In particular, $(z_1,\dots,z_n)$ will be a covering allocation for prices $p'$ and demand $d$.
By \Cref{thm:max_market_clearing_is_max_walrasian}, we can conclude that for the maximal Walrasian prices $p$ for demand $d$, we have $p' \leq p$.
\end{proof}

In contrast to the monotonicity in supply and demand, there is no monotonicity if we change the valuation functions for some buyers and some items.
We provide an example showing this for minimal Walrasian prices with purely additive valuation functions (which are strong gross substitutes):
\begin{example}
    Let there be three items $E = \{e_1, e_2, e_3\}$ and three buyers with unit-demand valuations $v_1 = (2, 3, 0)$, $v_2 = (0,1,1)$ and $v_3 = (0,1,1)$, i.e., $v_i(S) = \max_{e \in S} v_{ie}$ for all $S \subseteq E$.
    The buyer-optimal Walrasian prices are given by $\pmin = (0, 1, 1)$.
    Now, assume that the valuation of the first buyer for the second item $e_2$ is decreased by one, i.e., $v_1' = (2,2,0)$.
    In this setting, the prices $\pmin' = (0, 0, 0)$ is the buyer-optimal Walrasian price vector.
    Hence, the price of the items whose valuation was not changed are the same or reduced by one unit.

    Consider the same setting again and assume that the first buyer decreases her valuation for item $e_1$ as well, i.e., $v''_1 = (1, 2, 0)$.
    The minimal Walrasian price vector is $\pmin'' = (0, 1, 1)$.
    Thus, the price of the items whose valuation was not changed are increased by one unit.
\end{example}

The example shows that the minimal Walrasian prices do not change in a monotone way when the valuations change (an example for maximal Walrasian prices can be found in \Cref{appendix:maxWalrasian_valuation_monotone}).
This seems to be intuitive since if the valuation of one item is decreased, it could increase the demand for other items which then increases the price.
On the other hand, the item itself is less attractive, such that the demand can reduce (and thus also potential conflicts).

Note also that the question regarding price monotonicity w.r.t. changes in valuations is a somewhat very restricted problem, as the strong gross substitutes property might get lost if one is not careful with the changes in the valuation function.

\section{Conclusion}
We provide a simple, combinatorial algorithm to compute minimal maximal overdemanded sets and minimal maximal underdemanded sets for markets where items are available in multiple copies and all buyers have strong gross substitutes valuation functions.
The algorithm is essentially an implementation of a polymatroid sum algorithm using an appropriate oracle model.
It turns out to be quite useful as it allows for a fast execution of the dynamic auction step.

Moreover, we prove that unique minimal packing prices as well as unique maximal covering prices always exist.
The unique minimal packing prices coincide with the minimal Walrasian prices, while the unique maximal packing prices coincide with the maximal Walrasian prices.
The distinction between packing prices and Walrasian prices had been mostly overlooked in the literature.
The clear distinction made here and the aforementioned results allow us to prove monotonicity properties of minimal and maximal Walrasian prices w.r.t. changes in demand and supply.

\medskip
\noindent
\textbf{Acknowledgment.}\quad We thank Andr\'as Frank and Georg Loho for fruitful discussions and comments.

\noindent
\textbf{Data availability statement.} No data are associated with this article.
Data sharing is not applicable to this article.

\bibliographystyle{abbrvnat}
\bibliography{literature}

\appendix
\section*{Appendix}
\section{Examples}\label{app:examples}
\subsection{Walrasian prices may not exist if valuations are not strong gross substitutes}\label{app:Walrasian_prices_dont_exist_without_GS}
\begin{example}\label{example:nongs_nowalrasian}
    Consider an example with two buyers $N = \{1, 2\}$ and three items $E = \{e_1, e_2, e_3\}$.
    The multiplicity of every item is one, i.e., $\U(e) = 1$ for all $e \in E$.
    The valuation functions of the buyers are as follows (note that they are not strong gross substitutes):
    \begin{align*}
        v_1(\chi_S) = \begin{cases}
            2 & \text{if } \{e_1, e_2\} \subseteq S,\\
            1 & \text{if } \{e_1, e_2\} \not\subseteq S, \{e_3\} \subseteq S,\\
            0 & \text{otherwise,}
        \end{cases}
        \qquad
        v_2(\chi_S) = \begin{cases}
            2 & \text{if } \{e_2, e_3\} \subseteq S,\\
            1 & \text{if } \{e_2, e_3\} \not\subseteq S, \{e_1\} \subseteq S,\\
            0 & \text{otherwise.}
        \end{cases}
    \end{align*}
    Here $p = (0, 0, 0)$ is not packing since the unique minimal preferred bundle of buyer 1 is $(1, 1, 0)$, while the unique minimal preferred bundle of buyer 2 is $(0, 1, 1)$.
    The vector $p = (0, 1, 0)$ is packing but it is not covering, since item $e_2$ is not sold.

    Assume that there is a packing and covering allocation $\bz$, i.e. $\sum_{i \in \{1, 2\}} z_i(e) = 1$ for all $e \in E$ and $z_i \in \dset_i(p)$ for $i \in \{1,2\}$.
    \begin{itemize}[leftmargin=*]
        \item In case that $z_1 = (0, 0, 0)$, it holds that $p(e_1) + p(e_2) \geq 2$ and $p(e_3) \geq 1$, since $z_1 \in \dset_1(p)$.
            So, the utility for the second buyer is $v_2(\chi_E) - p(E) \leq 2 - 3 = -1$ and thus, $\chi_E = (1, 1, 1) \notin \dset_2(p)$.
        \item In case that $z_1 = (1, 0, 0) \in \dset_1(p)$, the utility of buyer $1$ needs to be zero and $p(e_1) = 0$, $p(e_2) \geq 2$ and $p(e_3) \geq 1$.
            Hence, the second buyer receives items $e_2$, $e_3$ with utility $v_2(\chi_{\{e_2, e_3\}}) - p(\{e_2, e_3\}) \leq 2 - 3 = -1$ which is a contradiction to $\chi_{\{e_2, e_3\}} \in \dset_2(p)$.
        \item If $z_1 = (0, 1, 0) \in \dset_1(p)$, the utility of buyer $1$ is zero again and $p(e_1) \geq 2$, $p(e_2) = 0$ and $p(e_3) \geq 1$.
            Then, buyer~2 does not receive a preferred bundle since the utility of $\chi_{\{e_1, e_3\}} = (1, 0, 1)$ is given by $v_2(\chi_{\{e_1, e_3\}}) - p(\{e_1, e_3\}) \leq 1 - 2 = -2$.
        \item In case that $z_1 = (0, 0, 1) \in \dset_1(p)$, we know that $p(e_3) \leq 1$ and $p(e_1) + p(e_2) \geq p(e_3) + 1$.
            Since $\chi_{\{e_1, e_2\}} = z_2 \in \dset_2(p)$, it holds that $p(e_2) = 0$ and $v_2(\chi_{\{e_1, e_2\}}) - p(\{e_1, e_2\}) \geq v_2(\chi_{\{e_2, e_3\}}) - p(\{e_2, e_3\})$ which is equivalent to $p(e_1) + 1 \leq p(e_3)$.
           This is a contradiction to $p(e_1) \geq p(e_3) + 1$.
    \end{itemize}
    In case that $\|z_1\|_1 \geq 2$, we obtain the same contradictions by using $\|z_2\|_1 \leq 1$ and the symmetry of the valuations.
    To sum up, there exists no Walrasian price vector $p$ in this market.
\end{example}

\subsection{Packing vectors do not form a lattice given non strong gross substitutes valuations}\label{appendix:packing_no_lattice}
In the following, we show that the component-wise minimal vector of packing prices is not necessarily packing if valuations are not strong gross substitutes.
Note that Walrasian prices are not guaranteed to exist if buyer valuations are not all strong gross substitutes \citep{kelso1982job}.
Buyer-optimal packing prices are neither unique nor covering in the following example.

\begin{example}\label{example:nogs_nolattice}
    Given a market with four items $E = \{e_1, e_2, e_3, e_4\}$ where $\U(e) = 1$ for each $e \in E$ and two buyers.
    The valuations for the first buyer are given by the item-wise valuations $\tilde{v}_1 = (6, 6, 6, 10)$ and the independent sets $\I_1 = \{I \subseteq F \mid F \in \{\{e_1, e_2, e_3\}, \{e_4\}\}\}$ (i.e., not a matroid).
    Thus, the valuation of a bundle $\chi_S$ is given by $v_1(S) = \max_{I \in \I_1} \sum_{e \in I \cap S} \tilde{v}_1(e)$.
    The valuation function for the second buyer is defined similarly using item-wise valuations $\tilde{v}_2 = (10, 6, 6, 6)$ and the independent sets $\I_2 = \{I \subseteq F \mid F \in \{\{e_1\}, \{e_2, e_3, e_4\}\}\}$.

    First we show that there is no packing price vector $p$ with $\sum_{e \in E} p(e) < 8$.
    Consider such a vector $p$ and let $E'$ be all elements in $e \in E$ with $p(e) < 6$.
    By definition, $\abs{E \setminus E'} \leq 1$.
    The unique preferred bundle for buyer~1 is $\{e_1, e_2, e_3\} \cap E'$ and for buyer~2, it is $\{e_2, e_3, e_4\} \cap E'$.
    Since $\abs{\{e_1, e_2, e_3\} \cap \{e_2, e_3, e_4\}} = 2$ the intersection of the preferred bundles is non-empty, so $p$ cannot be a packing price vector.

    Moreover, there are packing price vectors with $\sum_{e \in E} p(e) = 8$, namely $p = (0, 2+\alpha, 6-\alpha, 0)$ with $\alpha \in [0, 4]$ are packing price vectors, since the first buyer is indifferent between $\{e_1, e_2, e_3\}$ and $\{e_4\}$ and the second one is indifferent between $\{e_1\}$ and $\{e_2, e_3, e_4\}$.
    Thus, the allocation $\{e_4\}$ to the first buyer and $\{e_1\}$ to the second buyer is packing.
    This family of packing prices gives a component-wise minimal vector $(0, 2, 2, 0)$ which itself is not packing.
    However, since $\sum_{e \in E} p(e) = 8$ for all price vectors in the defined family, they are all buyer-optimal packing prices.
    Hence, clearly, the buyer-optimal packing prices are not unique and moreover these buyer-optimal packing price vectors are all not covering.
\end{example}

Note that the valuations functions in this example are not strong gross substitutes (in contrast to weighted matroid rank valuations the independent sets do not fulfill the exchange property and thus, do not form a matroid).
Consider the valuation function of the first buyer.
It holds that$ \{e_1, e_2, e_3\} \in D_1(0)$.
But if we increase the prices of items $e_1$ and $e_2$ by $6$, i.e., $p = (6, 6, 0, 0)$, there is no demand for $e_3$ anymore since $D_1(p) = \{\{e_4\}\}$.
Similarly, the valuation function of the second buyer is also not strong gross substitutes.

\subsection{Covering vectors do not form a lattice given non strong gross substitutes valuations}
\label{appendix:covering_no_lattice}

Here we give a example that the join of two packing price vectors is not necessarily packing if the valuation functions are not gross substitutes.
\begin{example}
    Given a market with three items $E=\{e_1,e_2,e_3\}$ and unit-supply, i.e, $b(e)=1$ for each $e \in E$, and two buyers.
	Both buyers have the same valuation function, given by the following table:
    \begin{center}
    \begin{tabular}{C{7ex}|C{7ex}C{7ex}C{7ex}C{7ex}C{7ex}C{7ex}C{7ex}C{7ex}C{1ex}}
        $S$ & $\emptyset$ & $\{e_1\}$ & $\{e_2\}$ & $\{e_3\}$ & $\{e_1,e_2\}$ & $\{e_1,e_3\}$ & $\{e_2,e_3\}$ & $\{e_1,e_2,e_3\}$ & \\
        \midrule
        $v_i(S)$ & 0 & 7 & 7 & 8 & 14 & 13 & 13 & 18
    \end{tabular}
    \end{center}
    Then $p=(6,7,7)$ is covering, since $\dset_i(p)=\{\{e_1\},\{e_1,e_2\},\{e_3\}\}$.
	Moreover, $p'=(6,7,7)$ is covering with $\dset_i(p)=\{\{e_2\},\{e_1,e_2\},\{e_3\}\}$.
    However, $q = p \vee p'=(7,7,7)$ is not covering, as $\dset_i(q)=\{\{e_3\}\}$.

    Note there is a Walrasian price vector above all these price vectors, which is $\pmax=(7,7,8)$.
\end{example}

\subsection{Maximal Walrasian prices are not monotone with respect to changes in the valuations}
\label{appendix:maxWalrasian_valuation_monotone}

As already seen for minimal Walrasian prices, the maximal Walrasian prices do not change in a monotone way when the valuations are changed:
\begin{example}
    Let there be four items $E=\{e_1,e_2,e_3\}$ and three buyers with unit-demand valuations i.e., $v_i(S) = \max_{e\in S} v_{ie}$.
	The valuations per item and the corresponding maximal Walrasian prices can be found in the following table:

    \vspace{1ex}
    \begin{center}
    \begin{tabular}{ccc|c}
        \toprule
        valuations buyer 1 & valuations buyer 2 & valuations buyer 3 & max Walrasian price \\
        $(v_{11},v_{12},v_{13},v_{14})$ & $(v_{21},v_{22},v_{23},v_{24})$ & $(v_{31},v_{32},v_{33},v_{34})$ & $\pmax$ \\
        \midrule
        $(0,9,1,1)$ & $(6,10,0,0)$ & $(4,0,1,1)$ & $(4,8,0,0)$\\
        $(0,9,2,1)$ & $(6,10,0,0)$ & $(4,0,1,1)$ & $(3,7,0,0)$\\
        $(0,9,2,2)$ & $(6,10,0,0)$ & $(4,0,1,1)$ & $(3,7,0,0)$\\
        $(0,10,2,2)$ & $(6,10,0,0)$ & $(4,0,1,1)$ & $(4,8,0,0)$\\
        \bottomrule
    \end{tabular}
    \end{center}
    Hence, also for the maximal Walrasian prices, the price can either stay the same, go up or go down, i.e., there is no monotone behavior.
\end{example}

\subsection{Example of running time of different oracle calls}\label{app:oracle_calls}
Let us consider an OXS valuation function $v_i$ in the unit-supply case.
This is described by a weighted bipartite graph $(E \cup T_i, H_i, w^{(i)})$.
For a vector $z_i \in \{0, 1\}^E$, we let $S_i = \supp(z_i)$.
Then, $v_i(z_i)$ is the maximum weight of a matching between the sets $S_i$ and $T_i$.

Thus, the time $\BO$ needed to find a set $z_i \in \DM_i(p)$ amounts to the time needed to compute a maximum weight matching.
This can be done using the Hungarian algorithm; better algorithms exist if the graph is not dense, or in the weakly polynomial regime.
However, given $z_i \in \DM_i(p)$ along with an optimal matching and an optimal dual solution, the time $\EO$ of a $w_i(e, f)$ query is the time of finding a single augmenting path.

\subsection{Example of an exchange graph}\label{apx:examples:exchange-graph}
Suppose there are six items $E = \{e_1, e_2, e_3, e_4, e_5, e_6\}$, every $e \in E$ is available in quantity $\U(e) = 2$, and there are three buyers $N = \{{\color{rwthblue}Blue}, {\color{rwthred}Red}, {\color{rwthgreen}Green}\}$.

Further suppose at some stage of the auction, the allocation $\bz = (z_\BB, z_\RR, z_\GG)$ is as follows:
\begin{align*}
    z_\BB &= (2, 2, 2, 1, 0, 0),\\
    z_\RR &= (2, 1, 0, 0, 1, 1),\\
    z_\GG &= (0, 2, 0, 1, 1, 0).
\end{align*}
Thus, the items $\{e_1, e_2\}$ are oversold and $\{e_6\}$ is undersold at current prices.
The auctioneer also queries each buyer $i \in N$ for the values $w_i(e, f)$ for each $e, f \in E$ with $z_i(f) > 0$ and $z_i(e) < b(e)$.
Suppose the buyers answer with the following non-zeros:
\begin{alignat*}{6}
    w_\BB(e_4, e_1) &= 1, \qquad    & w_\RR(e_3, e_1) &= 2, \qquad & w_\GG(e_1, e_5) &= 1,\\
    w_\BB(e_4, e_2) &= 1, \qquad    & w_\RR(e_3, e_2) &= 1, \qquad & w_\GG(e_3, e_4) &= 1,\\
    w_\BB(e_4, e_3) &= 1, \qquad    & w_\RR(e_5, e_6) &= 1, \qquad & w_\GG(e_4, e_2) &= 1,\\
                    &               &                 &            & w_\GG(e_4, e_5) &= 1.
\end{alignat*}
Then the exchange graph looks as depicted as in \Cref{fig:ex-exchange-graph}.
Note that there is no path from an undersold to an oversold item (otherwise the polymatroid sum problem would not have been solved to optimality).
The set of items from which we can reach an oversold item is $\{e_1, e_2, e_3, e_4\}$, which then by \Cref{thm:computation_overd-set} is a minimal maximal overdemanded set (i.e., not just the oversold items themselves but also suitable substitutes for them).
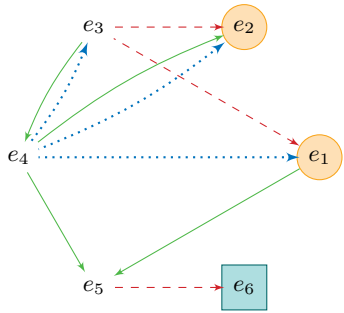
\begin{figure}
    \centering
    \begin{tikzpicture}
        \node (1) at (0:2) {$e_1$};
        \node (2) at (60:2) {$e_2$};
        \node (3) at (120:2) {$e_3$};
        \node (4) at (180:2) {$e_4$};
        \node (5) at (240:2) {$e_5$};
        \node (6) at (300:2) {$e_6$};

        \draw[<-,rwthblue,dotted,thick] (1) to (4);
        \draw[<-,rwthblue,dotted,thick] (2) to[bend left=8] (4);
        \draw[<-,rwthblue,dotted,thick] (3) to[bend left=8] (4);

        \draw[<-,rwthred,dashed] (1) to (3);
        \draw[<-,rwthred,dashed] (2) to (3);
        \draw[<-,rwthred,dashed] (6) to (5);

        \draw[<-,rwthgreen,solid] (4) to[bend left=8] (3);
        \draw[<-,rwthgreen,solid] (2) to[bend right=8] (4);
        \draw[<-,rwthgreen,solid] (5) to (1);
        \draw[<-,rwthgreen,solid] (5) to (4);

        \begin{pgfonlayer}{background}
			\tube{(1)}{rwthorange,fill=rwthorange!30}{17}
			\tube{(2)}{rwthorange,fill=rwthorange!30}{17}
            \ubox{(6)}{rwthpetrol,fill=rwthturquoise!30}{17}
		\end{pgfonlayer}
    \end{tikzpicture}
    \caption{Exchange graph $G(\bz) = (E, A)$ where blue (dotted) lines indicate $A_\BB$, red (dashed) lines indicate $A_\RR$, and green (solid) lines indicate $A_\GG$.
	Oversold items are indicated by orange cycles, undersold items are indicated by turquoise squares.}
    \label{fig:ex-exchange-graph}
\end{figure}

\newpage
\section{Missing proofs}\label{apx:missing-proofs}

\subsection{Packing, Covering and Walrasian}
\label{apx:packing_covering}

A Walrasian price vector is packing and covering.
Now, we show the other direction, i.e., that $p$ is Walrasian if it is packing and covering.

\begin{lemma}
    \label{lem:packing_covering_Walrasian}
	Let $p$ be a price vector in an auction with strong gross substitutes valuations that is both packing and covering.
	Then $p$ is Walrasian.
\end{lemma}
\begin{proof}
    Let $\bz = (z_1, \dots, z_n)$ be a covering allocation with respect to $p$.
    Further, let $\mathbf{y} = (y_1, \dots, y_n)$ be a packing allocation with respect to $p$ that, among all such collections, minimizes $\sum_{i \in N} \sum_{e \in E} |z_i(e) - y_i(e)|$.
    Note that such allocations exists as $p$ is both packing and covering.
    We will show that we have $\sum_{i \in N} y_i(e) = b(e)$ for every $e \in E$, i.e., $\mathbf{y}$ is also covering and thus $p$ is Walrasian.

    Assume towards a contradiction that this is not the case, i.e., there is an $e \in E$ such that $\sum_{i \in N} y_i(e) \leq b(e)-1$ (where we used that $b$ and all $y_i$ are integral).
    As $ b(e) \leq \sum_{i \in N} z_i(e)$, we know that there is a $j \in N$ such that $y_j(e) < z_j(e)$.
    As $v_j$ is gross substitutes, it is \Mnat-concave by \Cref{lemma:gs_is_Mnat-convave}.
    Hence, $u_j$ is \Mnat-concave too, as the difference of an \Mnat-concave function and a modular function.
    As $e \in \supp^+(z_j-y_j)$, we can infer that one of the following holds:
    \begin{enumerate}[label=(M\arabic*),leftmargin=*]
        \item $u_j(z_j) + u_j(y_j) \leq u_j(z_j-\chi_e) + u_j(y_j+\chi_e)$, or\label{M1}
        \item there exists $f \in \supp^-(z_j-y_j)$ such that $u_j(y_j) + u_j(z_j) \leq u_j(z_j-\chi_e+\chi_f) + u_j(y_j+\chi_e-\chi_f)$.\label{M2}
    \end{enumerate}
    First, assume that \ref{M1} holds true.
    Then $y_j, z_j \in \dset_j(p)$ implies that both $z_j-\chi_e$ and $y_j+\chi_e$ are contained in $\dset_j(p)$ as well.

    Define an allocation $\mathbf{y}' = (y'_1, \dots, y'_n)$ by setting $y'_i = y_i$ for all $i \in N \setminus \{j\}$ and $y'_j = y_j + \chi_e$.
    Then $y'_i \in \dset_i(p)$ holds for all $i \in N$.
    Furthermore, we have $\sum_{i \in N} y'_i(g) = \sum_{i \in N} y_i(g) \leq b(g)$ for all $g \in E \setminus \{e\}$, and $\sum_{i \in N} y'_i(e) = \sum_{i \in N} y_i(e) + 1 \leq b(e)$.
    Moreover, we have $|z_i(g) - y'_i(g)| = |z_i(g) - y_i(g)|$ for all $(i, g) \in N \times E \setminus \{(j, e)\}$, and $|z_j(e) - y'_j(e)| = |z_j(e) - y_j(e)| - 1$ because $e \in \supp^+(z_j-y_j)$.
    But this contradicts our assumption that $\mathbf{y}$ minimizes $\sum_{i \in N} \sum_{g \in E} |z_i(g)-y_i(g)|$.

    Next, assume that \ref{M2} holds.
    Again, $y_j, z_j \in \dset_j(p)$ allows us to conclude that both $z_j-\chi_e+\chi_f$ and $y_j+\chi_e-\chi_f$ are contained in $\dset_j(p)$ as well.
    Define a collection of bundles $y'$ by setting $y'_i = y_i$ for all $i \in N \setminus \{j\}$, and $y'_j = y_j+\chi_e-\chi_f$.
    Then $y'_i \in \dset_i(p)$ for all $i \in N$.
    Moreover, $\sum_{i \in N} y'_i(g) = \sum_{i \in N} y_i(g) \leq b(g)$ for all $g \in E \setminus \{e,f\}$, $\sum_{i \in N} y'_i(e) = \sum_{i \in N} y_i(e) + 1 \leq b(e)$ and $\sum_{i \in N} y'_i(f) = \sum_{i \in N} y_i(f) - 1 \leq b(f)$.
    Additionally, for $(i, g) \in N \times E \setminus \{(j, e), (j, f)\}$, we have $|z_i(g) - y'_i(g)| = |z_i(g) - y_i(g)|$.
    As $e \in \supp^+(z_j-y_j)$ and $f \in \supp^-(z_j-y_j)$, we further know that $|z_j(e) - y'_j(e)| = |z_j(e) - y_j(e)| - 1$ and $|z_j(f) - y'_j(f)| = |z_j(f) - y_j(f)| - 1$.
    But again, this contradicts the fact that $y$ minimizes $\sum_{i \in N} \sum_{e \in E} |z_i(e)-y_i(e)|$.

    To sum up, we obtained a contradiction in both cases and hence, we have $\sum_{i \in N} y_i(e) = b(e)$ for all $e \in E$.
    Thus, $\mathbf{y}$ is packing and covering.
    We conclude that $p$ is a Walrasian price vector.
\end{proof}

\subsection{Gross substitutes valuations yield \M-convex demand sets}\label{apx:missing-proofs:GS-Mconv}
\LemGSMconv*
\begin{proof}
    By \Cref{lemma:gs_is_Mnat-convave}, the valuation function $v_i$ is \Mnat-concave.
    Thus, the utility function $u_i$ is also \Mnat-concave as the sum of an \Mnat-concave and a modular function.
    Let $x, y \in \DM_i(p)$ and let $e \in \supp^+(x-y)$.
    Then, since $u_i$ is \Mnat-concave, either $u_i(x) + u_i(y) \leq u_i(x - \chi_e) + u_i(y + \chi_e)$, or there is $f \in \supp^-(x-y)$ with $u_i(x) + u_i(y) \leq u_i(x - \chi_e + \chi_f) + u_i(y + \chi_e -\chi_f)$.
    In the first case, $x - \chi_e, y + \chi_e \in \dset_i(p)$, and in the second case, $x - \chi_e + \chi_f, y + \chi_e - \chi_f \in \dset_i(p)$.
    The first case is impossible since $\DM_i(p)$ is the set of inclusion wise minimal vectors in $\dset_i(p)$, but $x - \chi_e \leq x$.
    Hence, the second case holds and it remains to show that there is no preferred bundle strictly included in $x' \coloneqq x - \chi_e + \chi_f$ or $y' \coloneqq y + \chi_e - \chi_f$.

    W.l.o.g.\ consider $x'$ and assume that there is a preferred bundle $z \neq x'$ with $z \leq x'$.
	Hence, there exist an item $g$ with $z(g) < x'(g) \leq x(g)$.
    As $x$ is a minimal preferred bundle, we must have $z(f)=x(f)+1$, as otherwise $z \leq x$ as well.
    Moreover, $z(e) \leq x(e)-1$ as $z \leq x'$.
    Hence, $f \in \supp^+(z-x)=\{f\}$ and thus by \Mnat-concavity of $u_i$, we have $u_i(z) + u_i(x) \leq u_i(x + \chi_f) + u_i(z - \chi_f)$ or there exists an item $g \in \supp^-(z-x)$ with $u_i(z) + u_i(x) \leq u_i(x + \chi_f - \chi_g) + u_i(z - \chi_f + \chi_g)$.
	By minimality of $x$, the first case is not possible since $z - \chi_f \leq x$ with $z-\chi_f \neq x$.

    Thus, we can focus on the second case with $z - \chi_f + \chi_g \in \dset_i(p)$.
	By definition of $z$ and $g$, we know that $z-\chi_f+\chi_g \leq x$.
	Hence, by minimality of $x$ it follows that $z-\chi_f+\chi_g = x$.
	As $z(e)<x(e)$, this implies $g=e$.
	However, this yields $z(e)+1\leq x'(e)<x(e)$, a contradiction.

    Hence, $x - \chi_e + \chi_f, y + \chi_e - \chi_f \in \DM_i(p)$ holds and we have shown that $\DM_i(p)$ is an \M-convex set.

    The proof that $\dhat_i(p)$ is an M-convex set follows the same lines.
\end{proof}

\subsection{Underdemanded sets}
\label{app:underdemanded_sets}

\LemCovering*
\begin{proof}
    Observe that for any price vector $p$, we have
    \begin{equation}
        \label{eq:proof_no_underdemanded_set}
        \max_{y \in \dhat(p)} \left\{ \sum_{e \in E} \min \left\{ \sum_{i \in N} y_i(e), b(e) \right\} \right\} \leq \sum_{e \in E} b(e).
    \end{equation}
    Moreover, equality holds in \eqref{eq:proof_no_underdemanded_set} if and only if there is a covering allocation $z \in \dhat(p)$, i.e., $p$ is covering.
    Now, assume that $p$ is covering.
    Using the Min-Max Theorem for polymatroid sum (\Cref{thm:poly-union}), we get
    \begin{equation*}
        \min_{T \subseteq E} \left\{ b(T) + \sum_{i \in N} \rkmax(E\setminus T) \right\}
        = \max_{y \in \dhat(p)} \left\{ \sum_{e \in E} \min \left\{ \sum_{i \in N} y_i(e), b(e) \right\} \right\}
        = \sum_{e \in E} b(e).
    \end{equation*}
    We obtain
    \begin{align*}
        && b(T) + \sum_{i \in N} \rkmax(E\setminus T) &\geq b(E) \qquad &\text{for all } T \subseteq E\\
        \Leftrightarrow \qquad && b(E \setminus S) + \sum_{i \in N} \rkmax(S) &\geq b(E) \qquad &\text{for all } S \subseteq E\\
        \Leftrightarrow \qquad && 0 &\geq \underd^p(S) &\text{for all } S \subseteq E
    \end{align*}
    i.e., there is no underdemanded set $S$.

    Next we show that there is a covering allocation if there is no underdemanded set.
    By the same arguments as above, we get
    \begin{equation}
        \label{eq:proof_covering_allocation}
        b(E) \leq \min_{T \subseteq E} \left\{ b(T) + \sum_{i \in N} \rkmax(E\setminus T) \right\}
        = \max_{y \in \dhat} \left\{ \sum_{e \in E} \min \left\{ \sum_{i \in N} y_i(e), b(e) \right\} \right\}
        \leq \sum_{e \in E} b(e),
    \end{equation}
    where the first inequality holds as there is no underdemanded set $S$ and the next equality holds by the Min-Max Theorem for polymatroid sum (\Cref{thm:poly-union}).
    Thus, \eqref{eq:proof_covering_allocation} holds with equality everywhere, especially for the last equation.
	But this is only possible if there is a covering allocation $z \in \dhat$.
	Hence, the prices $p$ are covering if there is no underdemanded set.
\end{proof}

\LemMinMaxUnder*
\begin{proof}
    Let $S$ be a maximal underdemanded set and let $T\coloneqq E\setminus S$.
    By \Cref{lemma:max_underdemanded_iff}, \Cref{thm:poly-union} and optimality of $\bz$, we have
    \[
        \sum_{e \in E} \min \left\{\sum_{i \in N} z_i(e), \U(e)\right\} = \sum_{i \in N} \rkmax^p_i(E\setminus T) + \U(T).
    \]
    By \Cref{lem:poly-union-opt}, $S = E \setminus T$ satisfies all three properties stated in this lemma.
    Moreover, if two sets $S$ and $S'$ satisfy these three properties, then $S \cap S'$ and $S \cup S'$ also satisfy them.
    Hence, there exists a unique minimal set $S$ satisfying these three properties.
    By \Cref{lem:poly-union-opt}, $E\setminus S$ is a minimizer of $T\mapsto \sum_{i \in N} \rkmax^p_i(E\setminus T) + \U(T)$, and for every other minimizer $T'$, we have $S\subseteq E\setminus T'$.
    Hence, $S$ is the minimal maximal underdemanded set for prices $p$.
\end{proof}

\ThmCompUnderdSet*
\begin{proof}
    We first show that $R \subseteq S$.
    By \Cref{lem:minmax-under}, $S$ includes all undersold items.
    Moreover, $S$ is $i$-tight for every $i \in N$, which, by \Cref{lem:tight-sets-weights}, means that whenever $e \in S$ and $(e, f) \in A$, we have $f \in S$ as well.
    Hence, $R \subseteq S$.

    As $S$ does not include any oversold item by \Cref{lem:minmax-under}, $R$ does not include any oversold item, either.
    By definition, $R$ contains all undersold items and $R$ does not have any leaving arc in $\widehat{G}(\bz)$.
    By \Cref{lem:tight-sets-weights}, this implies that $R$ is $i$-tight for every $i \in N$.
    Hence, $R$ satisfies all three conditions listed in \Cref{lem:minmax-under}.
    Minimality of $S$, thus, allows us to conclude that $S = R$.
\end{proof}

\subsection{Maximal overdemanded [underdemanded] sets are steepest descent directions of the Lyapunov function}\label{apx:missing-proofs:max-overd-steep}
\citet{ben2017walrasian} provided two alternative definitions for gross substitutes, adding to the standard one from Definition~\ref{def:GS} and the one from discrete convexity (see \citep{fujishige2003note} and \citep{murota2013computing}) as stated in Lemma~\ref{lemma:gs_is_Mnat-convave}.

Ben-Zwi's definition is in terms of change in a buyer's indirect utility function.
We add this lemma here to have a similar definition for the multi-supply setting, which Ben-Zwi did not consider.
The lemma holds in the same way as in the unit-supply setting by using Lemma~\ref{lem:Walrasian-reduction}.
\begin{lemma}\label{lem:BenGS}
    Given a valuation function $v_i\colon \rng \to \R_+$, the following are equivalent:
    \begin{enumerate}[ref=(\arabic*)]
        \item\label{lem:BenGS:2} For all $p \in \R_+^E$ and $S \subseteq E$, $V_i(p) = V_i(p+\chi_S) + \check{\mr}_i^p(S)$.
        \item\label{lem:BenGS:1} For all $p \in \R_+^E$ and $S \subseteq E$, $V_i(p) = V_i(p-\chi_S) - \rkmax_i^p(S)$.
        \item\label{lem:BenGS:3} The function $v_i$ is non-decreasing and gross substitutes.
    \end{enumerate}
\end{lemma}

With this lemma, we can prove that the steepest descent directions coincide either with the maximal overdemanded sets or with the maximal underdemanded sets.

\begin{lemma}[c.f.\ \citep{ben2017walrasian}]\label{lem:max-overd-is-steepest-descent}
    Let $p \geq 0$ be a price vector, then
    \[
        \arg\min \{ L(p + \chi_S) \mid S \subseteq E \} = \arg\max \{\overd^p(S) \mid S \subseteq E \}.
    \]
\end{lemma}
\begin{proof}[Proof of Lemma~\ref{lem:max-overd-is-steepest-descent}]
    We show that for all $S,T\subseteq E$, we have
    \begin{equation}L(p+\chi_T)-L(p+\chi_S)=\overd^p(S)-\overd^p(T)\label{eq:steepest_descent_overdemanded}.\end{equation}

    This implies the statement of the lemma because for a set $S^*\subseteq E$, we have
    \begin{align*}
    	S^*\in \arg\min \{ L(p + \chi_S) \mid S \subseteq E \} &\Leftrightarrow \forall T\subseteq E\colon L(p+\chi_T)-L(p+\chi_{S^*})\geq 0\\
    	\stackrel{\eqref{eq:steepest_descent_overdemanded}}{\Leftrightarrow} \forall T\subseteq E\colon \overd^p(S^*)-\overd^p(T)\geq 0 &\Leftrightarrow S^*\in\arg\max \{\overd^p(S) \mid S \subseteq E \}.\end{align*}

    To verify \eqref{eq:steepest_descent_overdemanded}, we calculate
    \begin{align*}
    	&\quad L(p+\chi_T) - L(p+\chi_S)\\
    	&= \sum_{i \in N} V_i(p+\chi_T) + \pr{p+\chi_T}{b} - \sum_{i \in N} V_i(p+\chi_S) - \pr{p+\chi_S}{b}\\
    	&= \sum_{i \in N} \left(V_i(p+\chi_T) - V_i(p+\chi_S)\right) + \pr{p+\chi_T}{b} - \pr{p+\chi_S}{b}\\
    	&= \sum_{i \in N} \left((V_i(p) -\check{\mr}_i^p(T)) - (V_i(p) - \check{\mr}_i^p(S))\right) + \pr{\chi_T-\chi_S}{b}\\
    	&= \sum_{i \in N} (\check{\mr}_i^p(S) - \check{\mr}_i^p(T)) + \pr{\chi_T-\chi_S}{b}\\
    	&= \left(\sum_{i \in N} \check{\mr}_i^p(S)-\pr{\chi_S}{b}\right) - \left(\sum_{i \in N} \check{\mr}_i^p(T)-\pr{\chi_T}{b}\right)\\
    	&= \overd^p(S) - \overd^p(T).\qedhere
    \end{align*}
\end{proof}

Analogously, we can also show the following result.
\begin{lemma}[c.f.\ \citep{ben2017walrasian}]\label{lem:max-underd-is-steepest-descent}
	Let $p \geq 0$ be a price vector, then
	\[
	\arg\min \{ L(p - \chi_S) \mid S \subseteq E \} = \arg\max \{\underd^p(S) \mid S \subseteq E \}.
	\]
\end{lemma}
\begin{proof}[Proof of Lemma~\ref{lem:max-underd-is-steepest-descent}]
We show that for all $S,T\subseteq E$, we have
\begin{equation}L(p-\chi_T)-L(p-\chi_S)=\underd^p(S)-\underd^p(T)\label{eq:steepest_descent_underdemanded}.\end{equation}

This implies the statement of the lemma because for a set $S^*\subseteq E$, we have
\begin{align*}
	S^*\in \arg\min \{ L(p - \chi_S) \mid S \subseteq E \} &\Leftrightarrow \forall T\subseteq E\colon L(p-\chi_T)-L(p-\chi_{S^*})\geq 0\\
	 \stackrel{\eqref{eq:steepest_descent_underdemanded}}{\Leftrightarrow} \forall T\subseteq E\colon \underd^p(S^*)-\underd^p(T)\geq 0 &\Leftrightarrow S^*\in\arg\max \{\underd^p(S) \mid S \subseteq E \}.\end{align*}

To verify \eqref{eq:steepest_descent_underdemanded}, we calculate
 \begin{align*}
	&\quad L(p-\chi_T) - L(p-\chi_S)\\
	&= \sum_{i \in N} V_i(p-\chi_T) + \pr{p-\chi_T}{b} - \sum_{i \in N} V_i(p-\chi_S) - \pr{p-\chi_S}{b}\\
	&= \sum_{i \in N} \left(V_i(p-\chi_T) - V_i(p-\chi_S)\right) + \pr{p-\chi_T}{b} - \pr{p-\chi_S}{b}\\
	&= \sum_{i \in N} \left((V_i(p) +\rkmax_i^p(T)) - (V_i(p) + \rkmax_i^p(S))\right) + \pr{\chi_S-\chi_T}{b}\\
	&= \sum_{i \in N} (\rkmax_i^p(T) - \rkmax_i^p(S)) + \pr{\chi_S-\chi_T}{b}\\
	&= \left(\pr{\chi_S}{b}-\sum_{i \in N} \rkmax_i^p(S) \right) - \left(\pr{\chi_T}{b}-\sum_{i \in N} \rkmax_i^p(T)\right)\\
	&= \underd^p(S) - \underd^p(T).\qedhere
\end{align*}
\end{proof}

\end{document}